\documentclass{elsarticle}
\usepackage{lineno,hyperref}
\usepackage{amsmath,amssymb,amsthm}
\usepackage{bm}
\usepackage{color}
\modulolinenumbers[5]

\journal{Journal of \LaTeX\ Templates}

\newtheorem{theorem}{Theorem}[section]
\newtheorem{define}[theorem]{Definition}
\newtheorem{remark}[theorem]{Remark}
\newtheorem{prop}[theorem]{Proposition}
\newtheorem{corollary}[theorem]{Corollary}
\newtheorem{lemma}[theorem]{Lemma}

\newtheorem{notation}[theorem]{Notation}

\newcommand{\C} {\mathbb{C}}
\newcommand{\Q} {\mathbb{Q}}
\newcommand{\Z} {\mathbb{Z}}

\newcommand{\bP} {\mathbb{P}}

\newcommand{\tdeg}{{\rm tdeg}}
\newcommand{\bfm}{{\bf m}}
\newcommand{\bfn}{{\bf n}}
\newcommand{\bfc}{{\bf c}}
\newcommand{\bfd}{{\bf d}}
\newcommand{\bfa}{{\bf a}}
\newcommand{\bfb}{{\bf b}}
\newcommand{\bfj}{{\bf j}}

\newcommand{\bfq}{{\bf q}}
\newcommand{\bfy}{{\bf y}}
\newcommand{\bfx}{{\bf x}}
\newcommand{\bfw}{{\bf w}}
\newcommand{\frakp}{{\mathfrak p}}

\newcommand{\calL}{{\mathcal L}}

\newcommand{\frakP}{{\mathfrak P}}

\newcommand{\frakh}{{\mathfrak h}}
\newcommand{\frakm}{{\mathfrak m}}
\newcommand{\VX}{{\vec X}}

\newcommand{\ord}{{\rm ord}}
\newcommand{\divs}{{\rm div}}

\newcommand{\supp}{{\rm supp}}
\newcommand{\rank}{{\rm rank}}
\newcommand{\res}{{\rm res}}

\newcommand{\overk}{{\overline{k(t)}}}
\begin{document}

\begin{frontmatter}

\title{Quasi-equivalence of heights in algebraic function fields of one variable}
\tnotetext[mytitlenote]{This work was supported by NSFC under Grants No.11771433 and No.11688101, by Beijing Natural Science Foundation (Z190004), by National Key Research and Development Project 2020YFA0712300, and by the Fundamental Research Funds for the Central Universities.}

\author{Ruyong Feng$^{b,a}$, Shuang Feng$^{a}$ and Li-Yong Shen$^{a}$}
\address{a. School of Mathematical Sciences, University of Chinese Academy of Sciences, 100049, Beijing, China\\
b. KLMM, Academy of Mathematics and Systems Science, Chinese Academy of Sciences, 100190, Beijing, China\\
ryfeng@amss.ac.cn, fengshuang@ucas.ac.cn, lyshen@ucas.ac.cn}

\begin{abstract}
For points $(a,b)$ on an algebraic curve over a field $K$ with height $\frakh$, the asymptotic relation between $\frakh(a)$ and $\frakh(b)$ has been extensively studied in diophantine geometry. When $K=\overk$ is the field of algebraic functions in $t$ over a field $k$ of characteristic zero, Eremenko in 1998 proved the following quasi-equivalence for an absolute logarithmic height $\frakh$ in $K$: Given $P\in K[X,Y]$ irreducible over $K$ and $\epsilon>0$, there is a constant $C$ only depending on $P$ and $\epsilon$ such that for each $(a,b)\in K^2$ with $P(a,b)=0$,
$$
   (1-\epsilon)\deg(P,Y)\frakh(b)-C\leq \deg(P,X)\frakh(a)\leq (1+\epsilon)\deg(P,Y)\frakh(b)+C.
$$
In this article, we shall give an explicit bound for the constant $C$ in terms of the total degree of $P$, the height of $P$ and $\epsilon$. This result is expected to have applications in some other areas such as symbolic computation of differential and difference equations.
\end{abstract}

\begin{keyword}
\texttt{height,\,algebraic curve,\,Riemann-Roch space}
\MSC[2020] 14Q05\sep  68W30
\end{keyword}

\end{frontmatter}

\section{Introduction}
In diophantine geometry, heights are often used to express the discreteness of algebraic points on an algebraic variety. They play an important role in diophantine geometry as well as other areas such as the theory of transcendence numbers. The study of the functorial property of heights can be tracked back to the date of Siegel, who gave the first asymptotic estimate of $\frakh(f(\bfc))$ in item of $\deg(f)$ when $\frakh(\bfc)$ is large enough, where $\bfc$ is a point on a projective algebraic curve and $f$ is a nonconstant rational function on this curve. Later, Siegel's result was improved by many authors (see for example N\'eron\cite{Neron}, Bombieri\cite{Bombieri}, Habegger\cite{Habegger-thesis}, Abouzaid\cite{Abouzaid} and Bartolome\cite{Bartolome}) who gave error terms to the asymptotic estimates. For instance, in \cite{Neron}, N\'eron proved the following quasi-equivalence of heights: Let $P\in \overline{\Q}[X,Y]$ be irreducible with $m=\deg(P,X)\geq 1$ and $n=\deg(P,Y)\geq 1$, then there is a constant $c(P)$ such that if $(a,b)\in \overline{\Q}^2$ with $P(a,b)=0$, the bound
$$
   \left|\frac{\frakh(a)}{n}-\frac{\frakh(b)}{m}\right|\leq c(P)\sqrt{\max\left\{\frac{\frakh(a)}{n},\frac{\frakh(b)}{m}\right\}}.
$$
An explicit estimate of the constant $c(P)$ is of particular interest in an effective version of Runge's theorem on the integer solutions of certain diophantine equations.
In \cite{Habegger-paper}, Habegger gave an explicit bound for the constant $c(P)$ and applied this bound to Runge's theorem. Other related height estimates may also be found in \cite{Abouzaid,Bartolome}.

The heights appearing in the above results are all defined in algebraic number fields. As to an absolute logarithmic height defined in function fields (see Section 3 of Chapter 3 in \cite{lang} for definition), Eremenko in 1998 proved quasi-equivalence of the following type.
\begin{prop}[Lemma 2 of \cite{eremenko}]
\label{prop:eremenko}
Let $P\in \overline{k(t)}[X,Y]$ be an irreducible polynomial of degree $m$ with respect to $X$ and of degree $n$ with respect to $Y$. Given $\epsilon>0$ there exists a constant $C$ depending on $P$ and $\epsilon$ such that for every $a,b\in \overk$ satisfying $P(a,b)=0$ we have
$$
   (1-\epsilon)n\frakh(b)-C\leq m\frakh(a)\leq (1+\epsilon)n\frakh(b)+C.
$$
\end{prop}
  One can see from Remark~\ref{rem:heights} that if $a\in k(t)$ then $\frakh(a)$ defined in the above proposition is exactly the degree of $a$, i.e. the maximum of the degrees of the numerator and denominator of $a$. Eremenko applied the above result to show that rational solutions of a first order algebraic ordinary differential equation (AODE) $F=0$ are of degree not greater than a constant only depending on $F$. From the viewpoint of algorithms, an explicit estimate of the constant $C$ is usually necessary to guarantee the termination of algorithms for computing rational solutions of AODEs. Meanwhile, such explicit estimate has potential applications in the algorithmic aspect of computing rational points on an algebraic variety over $k(t)$. In this article, we shall give an explicit bound for the constant $C$ in terms of the total degree of $P$, the height of $P$ and $\epsilon$. We obtain this explicit bound by computing the explicit expressions for constants appearing at each step of the proof of the above proposition given by Eremenko. In particular, we give bounds for the heights of the coefficients of a certain nonzero element in the Riemann-Roch space of a divisor. Precisely, suppose that $L=K(x,y)$ is an algebraic function field of one variable over a field $K$, where $x$ is transcendental over $K$ and $y$ is algebraic over $K(x)$. Then each element $A \in L$ can be presented as a polynomial in $y$ with coefficients in $K(x)$, i.e. $A=\frac{1}{q(x)}\sum_{i=0}^{n-1}\sum_{j=0}^m a_{i,j}x^j y^i$ with $a_{i,j}\in K, q\in K[X]$ and $n=[L:K(x)]$. For a certain nonzero element $A$ in the Riemann-Roch space of a divisor, we give a bound for the height of the projective point $\bfa=(\dots:a_{i,j}:\cdots)$ (see Proposition~\ref{prop:heightsofRiemann-Roch}) as well as a bound for the height of $q(X)$.  Note that Schmidt in \cite{Schmidt:Constructionandestimationofbasesinfunctionfields} presented a bound for $m$, a degree bound for $q$ and a bound for the absolute values of the coefficients of the Puiseux series expansion of $A$ when $K$ is the field of algebraic numbers. Although it is possible to obtain a bound for the height of $\bfa$ by the results (mainly Theorem C2) presented in \cite{Schmidt:Constructionandestimationofbasesinfunctionfields}, we do not take this approach because the absolute logarithmic height under consideration in this paper satisfies the triangle inequality, i.e. $\frakh(a+b)\leq \frakh(a)+\frakh(b)$, which is not usually satisfied for absolute logarithmic heights defined in algebraic number fields. The triangle inequality enables us to obtain a simpler expression for the constant $C$. Finally, let us remark that the construction of the Riemann-Roch space of a divisor is one of the fundamental problems in the theory of algebraic function fields. Many algorithms have already been developed for this problem, see for example \cite{Abelard-Couvreur-Lecerf,Coates,Davenport,Hess,Huang-Ierardi,LeGluher-Spaenlehauer,Volcheck}.

The article is organized as follows. In Section 2, we introduce some basic concepts and notations about algebraic function fields of one variable and heights used in the later sections. In Section 3, we estimate the heights of the coefficients for a certain nonzero element in the Riemann-Roch space of a given divisor. Finally, in Section 4, we present an explicit bound for the constant $C$.

As usual, for a polynomial $P(X_1,\dots,X_m)$, we use $\tdeg(P)$ and $\deg(P,X_i)$ to denote the total degree of $P$ and the degree of $P$ with respect to $X_i$ respectively. $\bP^m(\cdot)$ denotes the projective space of dimension $m$ over a field and $(a_0:\dots:a_m)$ denotes a point in $\bP^m(\cdot)$ with coordinates $a_i$.

\section{Algebraic function fields of one variable and heights}\label{heights}
In this section, we will introduce some basic concepts and notations of algebraic function fields of one variable and heights. Readers are referred to \cite{ chevalley, lang, Serre, walker} for details.

\subsection{Algebraic function fields of one variable}
\label{sec:algebraicfunctionfield}
Throughout this subsection, $K$ always stands for an algebraically closed field of characteristic zero.
Let $L$ be an algebraic function field of one variable over $K$. Assume that $L=K(x,y)$ where $x$ is transcendental over $K$ and $y$ satisfies
$$
   P(x,y)=A_0(x)y^n+A_1(x)y^{n-1}+\cdots+A_n(x)=0,\,\,A_i\in K[X],
$$
where $P\in K[X,Y]$ is irreducible. Denote by $K((z))$ the quotient field of the ring of formal power series in $z$.
Let $\bfx(z)=\sum_{i=r}^{\infty} c_i z^i\in K((z))$ with $r\in \Z, c_i\in K, c_r\neq 0$, then $\ord_z(\bfx(z))$ is defined to be $r$.
We call $(\bfx(z),\bfy(z))\in K((z))^2$ a parametrization of $P(X,Y)=0$ provided $P(\bfx(z),\bfy(z))=0$ and $\bfx(z)$ or $\bfy(z)$ does not belong to $K$. If there is an integer $s\geq 2$ such that $\bfx(z),\bfy(z)\in K((z^s))$ then the parametrization $(\bfx(z),\bfy(z))$ is said to be reducible, otherwise irreducible. Two parametrizations $(\bfx(z),\bfy(z))$ and $(\tilde{\bfx}(z),\tilde{\bfy}(z))$ are said to be equivalent if there is $\bfw(z)\in K((z))$ with $\ord_z(\bfw(z))=1$
such that
$$\bfx(z)=\tilde{\bfx}(\bfw(z))\,\,\mbox{and}\,\,\bfy(z)=\tilde{\bfy}(\bfw(z)).$$
\begin{define}
\label{def:places}
An equivalent class of irreducible parametrizations is called a place of $P(X,Y)=0$.
\end{define}
It was shown on page 95 of \cite{walker} that an irreducible parametrization of $P(X,Y)=0$ is equivalent to the one of the type
\begin{equation}
\label{EQ:place}
(a+z^\mu, z^{\nu}(b_0+b_1z^{\ell_1}+\cdots))
\end{equation}
where $a\in K, b_i\in K\setminus\{0\}, \mu\in \Z\setminus\{0\},\nu,\ell_i\in \Z$, $0<\ell_1<\ell_2<\cdots$ and $\mu,\nu,\nu+\ell_1,\nu+\ell_2,\cdots$ have no common factor greater than 1, moreover if $\mu<0$ then $a=0$. In the rest of this article, all irreducible parametrizations of $P(X,Y)=0$ will be of the type (\ref{EQ:place}). Let $\frakp$ be a place of the form (\ref{EQ:place}). We say that $\frakp$ lies above $x-a$ if $\mu>0$, and lies above $1/x$ if $\mu<0$. The integer $|\mu|$ is called the ramification index of $\frakp$ with respect to $K(x)$, denoted by $e_{\frakp,K(x)}$. Suppose that $f\in L\setminus \{0\}$. The order of $f$ at $\frakp$, denoted by $\ord_\frakp(f)$, is defined to be $\ord_z(f(z^\mu+a,z^{\nu}(b_0+\cdots)))$. If $\ord_\frakp(f)>0$, $\frakp$ is called a zero of $f$ and if $\ord_\frakp(f)<0$, $\frakp$ is called a pole of $f$. It is well-known that a nonzero $f$ admits only finitely many zeros and poles. We make the convention to write $\ord_\frakp(0)=\infty$.
For $f,g\in L$, one can verify that
$$
\ord_{\frakp}(fg)=\ord_{\frakp}(f)+\ord_{\frakp}(g),\,\,\ord_{\frakp}(f+g)\geq \min\{\ord_{\frakp}(f),\ord_{\frakp}(g)\}
$$
where the equality in the last formula holds if $\ord_{\frakp}(f)\neq \ord_{\frakp}(g)$.

Denote $V(\frakp)=\{f\in L \mid \ord_\frakp(f)\geq 0\}$. One can check that $V(\frakp)$ is a discrete valuation ring of $L$. One can also check that for $f\in V(\frakp)$ there is a unique $c_f\in K$ such that $\ord_\frakp(f-c_f)>0$. We define a map $\pi_\frakp: V(\frakp)\rightarrow K$ given by $f\mapsto c_f$. Then $\pi_\frakp$ is a $K$-homomorphism. We make a convention with $\pi_\frakp(f)=\infty$ if $\ord_\frakp(f)<0$. The point $(\pi_\frakp(x),\pi_\frakp(y))$ is called the center of $\frakp$.
\begin{remark}
\label{rem:places}
In \cite{chevalley}, a place is presented by the unique maximal ideal of a discrete valuation ring of $L$ over $K$. Precisely, let $\frakp$ be a place of $P(X,Y)=0$ and let $V(\frakp)$ be as above.  Set $\frakm_\frakp=\{f\in L \mid \ord_\frakp(f)>0\}$. One sees that $\frakm_\frakp$ is the unique maximal ideal of $V(\frakp)$, which is the ``place'' defined in \cite{chevalley} corresponding to $\frakp$.  Conversely, given a discrete valuation ring $V$ of $L$ over $K$ with $\frakm$ as its unique maximal ideal, we can construct a unique place of $P(X,Y)=0$ corresponding to $V$.  Let $z$ be a uniformizing variable at $V$. Expanding $x,y$ as Puiseux series in $z$ yields an irreducible parametrization and thus a place of $P(X,Y)=0$. Furthermore, different choices of uniformizing variable at $V$ induce equivalent irreducible parametrizations and so the same place. Additionally, the order of $f$ at $\frakm_\frakp$ defined in \cite{chevalley} is nothing else but $\ord_\frakp(f)$.
\end{remark}
\begin{define}
\label{def:divisors}
A divisor $D$ of $L$ is a finite sum of places of $P(X,Y)=0$ with integer coefficients, i.e. $D=\sum_{\frakp} d_\frakp \frakp$ where $d_\frakp\in \Z$ and $d_\frakp=0$ for all but finitely many places $\frakp$. Specially, $D$ is called a zero divisor if all $d_\frakp=0$.
\end{define}
Suppose that $D=\sum_{\frakp} d_\frakp \frakp$ is a divisor of $L$.  We call $\sum_\frakp d_\frakp$, denoted by $\deg(D)$, the degree of $D$. The set of places $\frakp$ with $d_\frakp\neq 0$ is called the support of $D$, denoted by $\supp(D)$. We call $D$ an integral divisor if $d_\frakp\geq 0$ for all $\frakp$, denoted by $D\geq 0$.
For a nonzero $f\in L$, denote
$$
   \divs(f)=\sum_\frakp \ord_\frakp(f)\frakp,
$$
which is called the divisor of $f$. Each divisor $D$ can be uniquely written as $D^{+}-D^{-}$ where $D^{+}, D^{-}$ are integral divisors and $\supp(D^{+})\cap \supp(D^{-})=\emptyset$. Given a divisor $D$, denote
$$
   \calL_K(D)=\{f\in L \mid \divs(f)+D\geq 0\} \cup \{0\}.
$$
We call $\calL_K(D)$ the Riemann-Roch space of $D$ which is a $K$-vector space of finite dimension, and we denote its dimension by $\ell(D)$. By the Riemann-Roch theorem, $\ell(D)>0$ if $\deg(D)$ is not less than the genus of $L$ over $K$.

\subsection{Heights in an algebraic function field of one variable}
Throughout this subsection, $k(t)$ stands for the field of rational functions in $t$ with coefficients in an algebraically closed field $k$ of characteristic zero, and $\overk$ for the algebraic closure of $k(t)$. Let $L\subset \overk$ be a finite extension of $k(t)$. Then $L$ is an algebraic function field of one variable over $k$. Places in this subsection will be presented by maximal ideals of discrete valuation rings of $L$ over $k$ or equivalent classes of irreducible parameterizations of $P(X,Y)=0$, where $P(X,Y)=0$ is an algebraic curve whose function field coincides with $L$. Let us first define an absolute logarithmic height of a point in $\bP^m(\overk)$. Note that $\ord_\frakp(0)=\infty$.
\begin{define}
\label{def:height1}
Given $\bfa=(a_0: \dots :a_m)\in \bP^m(\overk)$, let $L$ be a finite extension of $k(t)$ containing all $a_i$. The {\em absolute logarithmic height} (or simply {\em height}) of $\bfa$, denoted by $\frakh(\bfa)$,  is defined to be
$$
 \frac{\sum_{\frakp} \max_{i=0}^m\{-\ord_{\frakp}(a_i)\}}{[L:k(t)]}
$$
where $\frakp$ ranges over all places of $L$ over $k$.
\end{define}
From Section 3 of Chapter 3 in \cite{lang}, $\frakh(\cdot)$ is a logarithmic height function. Actually it is an absolute logarithmic height function i.e. a logarithmic height function independent of the choices of the field $L$. To see this, let $\tilde{L}$ be a finite extension of $L$ and suppose that $\frakp$ is a place of $L$ over $k$. Then there are finitely many places $\frakP$ of $\tilde{L}$ over $k$ lying above $\frakp$, i.e. $\frakP\cap L=\frakp$. For brevity, denote by $\frakP | \frakp$ a place $\frakP$ lying above $\frakp$. Note that the relative degree of $\frakP_i$ is 1 because $k$ is algebraically closed, and due to Theorem 1 on page 52 of \cite{chevalley}, for a fixed $\frakp$, $\sum_{\frakP|\frakp} e_{\frakP,L}=[\tilde{L}:L]$. Moreover $\ord_{\frakP}(a)=e_{\frakP,L}\ord_\frakp(a)$ for any $a\in L$ and any $\frakP$ lying above $\frakp$. These imply that for a fixed $\frakp$,
\begin{align*}
\sum_{\frakP|\frakp}\max_i\{-\ord_\frakP(a_i)\}&=\sum_{\frakP|\frakp}\max_i\{-e_{\frakP,L}\ord_\frakp(a_i)\}=\sum_{\frakP|\frakp}e_{\frakP,L}\max_i\{-\ord_\frakp(a_i)\}\\
&=[\tilde{L}:L]\max_i\{-\ord_\frakp(a_i)\}.
\end{align*}
From this, one easily sees that $\frakh(\cdot)$ is independent of the choices of $L$. Using Definition~\ref{def:height1}, it is natural to define the height of an element in $\overk$ and a polynomial in $\overk[X_1,\dots,X_m]$ as follows.

\begin{define}
\label{def:height2}
\begin{enumerate}
\item For $a\in \overk$, we define the height of $a$ to be $\frakh((1:a))$, denoted by $\frakh(a)$.
\item Let $Q$ be a nonzero polynomial in $\overk[X_1,\dots,X_m]$. We define the height of $Q$ to be
\[
    \frakh(Q)=\begin{cases} 0 & \mbox{$Q$ contains exactly one term}\\ \frakh(\bfa) & \mbox{otherwise} \end{cases},
\]
where $\bfa$ is a point in some projective space whose coordinates are the coefficients of $Q$.
\end{enumerate}
\end{define}

In the following, we make a convention that $a/0=\infty$ for any $a\in \overk \setminus \{0\}$ and $\frakh(\infty)=0$.
\begin{remark}
\label{rem:heights}
Assume that $\bfa=(a_0:\dots:a_m) \in \bP^m(\overk)$.
\begin{enumerate}
\item
Suppose that $a_0=1$, then
 $$
 \frakh(\bfa)=\frac{\sum_{\frakp} \max\{0,-\ord_\frakp(a_1),\dots,-\ord_\frakp(a_m)\}}{[L:k(t)]}\geq 0.
 $$

\item
Let $a\in \overk$ and $L=k(t,a)$. Let $Q(X,Y)$ be a nonzero irreducible polynomial over $k$ such that $Q(t,a)=0$. It is clear that $\frakh(a)=0$ if $a\in k$. Now assume that $a\notin k$ and $\frakp_1,\dots,\frakp_s$ are all distinct poles of $a$ in $L$, then
\[
\frakh(a)=\frac{-\sum_{i=1}^{s}\ord_{\frakp_i}(a)}{[L:k(t)]}=\frac{[L:k(a)]}{[L:k(t)]}=\frac{\deg(Q,X)}{\deg(Q,Y)}.
\]
In particular, if $a\in k(t)$ then $\frakh(a)=\deg(a)$ which is defined to be the maximum of the degrees of the denominator and numerator of $a$.
\end{enumerate}
\end{remark}

The height given in Definition~\ref{def:height1} has the following properties.
\begin{prop}\label{prop:heightproperty1}
$\frakh(a^n)=\frakh(a^{-n})=n \frakh(a),\,a\in \overk \setminus \{0\},\,n\geq 0$.
\end{prop}
\begin{proof}
Let $L=k(t,a)$. For each place $\frakp$ of $L$ over $k$,
$$\max\{0,-\ord_\frakp(a^n)\}=\max\{0,-n\ord_\frakp(a)\}=n\max\{0,-\ord_\frakp(a)\}.$$
By deifnition, $\frakh(a^n)=n\frakh(a)$. For the first equality, it suffices to show that $\frakh(a)=\frakh(1/a)$. As $(1:a)=(1/a:1)$, one sees that
$$\frakh(a)=\frakh((1:a))=\frakh((1/a:1))=\frakh(1/a).$$
\end{proof}
Suppose that $ \Phi=(\phi_0:\dots:\phi_m)$ is a morphism, namely
\begin{align*}
  \Phi: \bP^{s_1}(\overk)\times \dots \times \bP^{s_r}(\overk)&\longrightarrow\bP^m(\overk)\\
                       \bfb &\longmapsto (\phi_0(\bfb):\dots:\phi_m(\bfb)),
 \end{align*}
 where $\phi_i\in \overk[X_{1,0},\dots,X_{1,s_1},\dots,X_{r,0},\dots,X_{r,s_r}]$ is a nonzero polynomial homogeneous in $X_{j,0},\dots, X_{j,s_j}$ of degree $d_j$ for all $j=1,\dots,r$. Write $\phi_i=\sum_{j=1}^s c_{i,j}\bfm_j$, where $c_{i,j}\in \overk$ and $\bfm_1,\dots,\bfm_s$ are all monomials in $X_{1,0},\dots$, $X_{1,s_1},\dots,X_{r,0},\dots,X_{r,s_r}$ of total degree $\sum_{j=1}^r d_j$.
 \begin{define}
 We define the height of $\Phi$, denoted by $\frakh(\Phi)$, to be
 $$
    \frakh((c_{0,1}:\dots:c_{i,j}:\dots:c_{m,s})).
 $$
 \end{define}
 The following proposition will play a key role in the rest of this paper. Although it is a trivial generalization of an existing result (see Proposition on page 15 of \cite{Serre} or Lemma 1.6 on page 80 of \cite{lang} for the case with $r=1$), we reprove this result for completeness and give an explicit estimate of the error term $c$ in the case of heights in algebraic function fields of one variable.
\begin{prop}\label{prop:heightproperty2}
Let $\Phi=(\phi_0:\dots:\phi_m)$ be as above.
Suppose $(\bfa_1,\dots,\bfa_r) \in \bP^{s_1}(\overk)\times\dots\times \bP^{s_r}(\overk)$ is a point on which $\Phi$ is defined. Then
$$
   \frakh(\Phi(\bfa_1,\dots,\bfa_r))\leq \sum_{i=1}^r d_i\frakh(\bfa_i)+\frakh(\Phi).
$$
\end{prop}
\begin{proof}
 Write $\bfa_i=(a_{i,0}:\dots:a_{i,s_i}), \bfj=(j_{1,0},\dots,j_{1,s_1},\dots,j_{r,0},\dots,j_{r,s_r})$ and
$$\phi_i=\sum_{\bfj} c_{i,\bfj}X_{1,0}^{j_{1,0}} \cdots X_{l,l'}^{j_{l,l'}}\cdots X_{r,s_r}^{j_{r,s_r}}$$
with $c_{i,\bfj}\in \overk$ and {$\sum_{l'=0}^{s_l} j_{l,\l'}=d_l$}.
Let $L$ be a finite extension of $k(t)$ containing all $a_{i,j}$ and all $c_{i,\bfj}$.
For each place $\frakp$ of $L$ over $k$, one has that
{\begin{align*}
    -\ord_\frakp(c_{i,\bfj} & a_{1,0}^{j_{1,0}}\cdots a_{l,l'}^{j_{l,l'}}\cdots a_{r,s_r}^{j_{r,s_r}})=-\ord_\frakp(c_{i,\bfj})-\sum_{l=1}^{r}\sum_{l'=0}^{s_l} j_{l,l'}\ord_\frakp(a_{l,l'})\\
    &\leq -\ord_\frakp(c_{i,\bfj})+\sum_{l=1}^{r}\sum_{l'=0}^{s_l} j_{l,l'}\max\{-\ord_\frakp(a_{l,0}),\dots,-\ord_\frakp(a_{l,s_l})\}\\
    &\leq \max_{i,\bfj}\{-\ord_\frakp(c_{i,\bfj})\}+\sum_{l=1}^r d_l \max\{-\ord_\frakp(a_{l,0}),\dots,-\ord_\frakp(a_{l,s_l})\}
    .
\end{align*}}
This implies that for each $i$,
\begin{align*}
  -\ord_\frakp&(\phi_i(\bfa_1,\dots,\bfa_r))\leq \max_{\bfj}\{-\ord_\frakp(c_{i,\bfj} a_{1,0}^{j_{1,0}}\dots a_{l,l'}^{j_{l,l'}}\dots a_{r,s_r}^{j_{r,s_r}})\}\\
  &\leq \max_{i,\bfj}\{-\ord_\frakp(c_{i,\bfj})\}+\sum_{l=1}^r d_l \max\{-\ord_\frakp(a_{l,0}),\dots,-\ord_\frakp(a_{l,s_l})\}.
\end{align*}
By definition, one sees that $\frakh(\Phi(\bfa_1,\dots,\bfa_r))\leq \sum_{l=1}^r d_l\frakh(\bfa_l)+\frakh(\Phi)$.
  \end{proof}
The above proposition has the following corollaries.
\begin{corollary}\label{cor:heightproperty11}
\begin{enumerate}
\item Suppose that $Q$ is a polynomial in $\Q[X_1,\dots,X_m]$ with degree $d_i$ in $X_i$ for all $i$. Let $b_1,\dots,b_m\in \overk$. Then $$\frakh(Q(b_1,\dots,b_m)\leq \sum_{i=1}^m d_i\frakh(b_i).$$
    Especially, $\frakh(a+b),\frakh(ab)\leq \frakh(a)+\frakh(b)$ for any $a,b\in \overk$.
\item Suppose that $c_1,c_2,c_3,c_4\in k$ satisfy that $c_1c_4-c_2c_3\neq 0$. Then $$\frakh\left(\frac{c_1 a+c_2}{c_3a+c_4}\right)=\frakh(a)$$ for any $a\in \overk$.
\end{enumerate}
\end{corollary}
\begin{proof}
1. Homogenizing $Q$, we obtain $\bar{Q}\in \Q[X_{1,0},X_{1,1},\dots,X_{m,0},X_{m,1}]\setminus \{0\}$ homogenous in $X_{i,0},X_{i,1}$ of degree $d_i$ for all $i$ such that
$$\bar{Q}(\bfb_1,\dots,\bfb_m)=Q(b_1,\dots,b_m)$$
 where $\bfb_i=(1:b_i)$. In Proposition~\ref{prop:heightproperty2}, if we take $\phi_0=\prod_{i=1}^m X_{i,0}^{d_i}, \phi_1=\bar{Q}, \bfa_i=\bfb_i$ then we have that
$$\frakh(Q(b_1,\dots,b_m))=\frakh(\bar{Q}(\bfb_1,\dots,\bfb_m))\leq \sum_{i=1}^m d_i\frakh(\bfb_i)=\sum_{i=1}^md_i\frakh(b_i).$$

2. If $c_3a+c_4=0$, then $c_1a+c_2 \neq 0$ since $c_1c_4-c_2c_3\neq 0$. One sees that $\frakh((c_1a+c_2)/(c_3a+c_4))=\frakh(\infty)=0$ and $\frakh(a)=\frakh(-c_4/c_3)=0$, then the desired equality holds. Now assume $c_3a+c_4\neq 0$, we take $r=1,s_1=1$ and
$$\Phi=(c_3X_{1,1}+c_4X_{1,0}, c_1X_{1,1}+c_2X_{1,0}), \,\,\bfa_1=(1:a)$$
in Proposition~\ref{prop:heightproperty2}. Then one has that
$$
  \frakh\left(\frac{c_1a+c_2}{c_3a+c_4}\right)=\frakh \left(\left(1:\frac{c_1a+c_2}{c_3a+c_4}\right)\right)=\frakh(\Phi(\bfa_1))\leq \frakh(a).
$$
Conversely, let $b=(c_1a+c_2)/(c_3a+c_4)$. Then $a=(c_4b-c_2)/(c_1-c_3b)$. A similar argument implies that $\frakh(a)\leq \frakh(b)$. Thus $\frakh(a)=\frakh(b)$.
\end{proof}

\begin{corollary}\label{cor:heightproperty12}
\begin{enumerate}
\item Let $\bfb_i=(b_{i,0}:\dots:b_{i,n_i})\in \bP^{n_i}(\overk)$ for $i=1,2$. Suppose that $b_{1,0}=b_{2,0}=1$ and set $$\bfc=(b_{1,0}:\dots:b_{1,n_1}:b_{2,0}:\dots:b_{2,n_2})\in \bP^{n_1+n_2+1}(\overk).$$
    Then $\frakh(\bfc)\leq \frakh(\bfb_1)+\frakh(\bfb_2)$.
\item Suppose that $\bfb=(b_0:\dots:b_n)\in \bP^n(\overk)$. Then $\frakh(\bfb)\leq \sum_{i=0}^n \frakh(b_i)$.
\end{enumerate}
\end{corollary}
\begin{proof}
1. We take $r=2,s_1=n_1,s_2=n_2$, $\Phi=(\phi_0:\dots:\phi_{n_1+n_2+1})$ with
\[
  \phi_i=
  \begin{cases}
  X_{1,i}X_{2,0} & i=0,\dots,n_1\\
  X_{1,0}X_{2,i-n_1-1} & i=n_1+1,\dots,n_1+n_2+1
  \end{cases}
\]
and $\bfa_1=\bfb_1,\bfa_2=\bfb_2$. Then $\Phi(\bfa_1,\bfa_2)=\bfc$ and $\frakh(\bfc)\leq \frakh(\bfb_1)+\frakh(\bfb_2)$ because of Proposition~\ref{prop:heightproperty2}.

2. In Proposition~\ref{prop:heightproperty2}, take $r=n+1, s_1=\dots=s_{n+1}=1$, $\Phi=(\phi_0,\dots,\phi_n)$ with
$\phi_i=X_{i+1,1}\prod_{j=1,j\neq i+1}^{n+1} X_{j,0}$, $\bfa_i=(1:b_{i-1})$ for all $i=1,\dots,n+1$.
\end{proof}

\begin{corollary}\label{cor:resultant}
Assume that $P_1,P_2\in \overk[X_1,\dots,X_m,Y]$. Then
$$
\frakh({\rm res}_{Y}(P_1,P_2))\leq \deg(P_2,Y)\frakh(P_1)+\deg(P_1,Y)\frakh(P_2)
$$
where ${\rm res}_{Y}(P_1,P_2)$ is the resultant of $P_1$ and $P_2$ with respect to $Y$.
\end{corollary}
\begin{proof}
The assertion is clear if ${\rm res}_{Y}(P_1,P_2)=0$. In the following, suppose that ${\rm res}_{Y}(P_1,P_2)\neq 0$.
Assume $\deg(P_i,Y)=n_i,\,i=1,2$. Denote $\VX=(X_1,\dots,X_m)$ and $\VX^\bfd=\prod_{i=1}^m X_i^{d_i}$ for $\bfd=(d_1,\dots,d_m)\in \Z^m$. Write
$$
P_1=\sum_{i=0}^{n_1} a_i(\VX)Y^i,\quad P_2=\sum_{i=0}^{n_2} b_i(\VX)Y^i
$$
where $a_i(\VX),b_j(\VX)\in \overk[\VX]$.
Then
$$
{\rm res}_Y(P_1,P_2)=
\begin{vmatrix}
a_{n_1}   & a_{n_1-1}  &\cdots  & a_0   \\
         & \ddots    &  \ddots       &   & \ddots   \\
         &           &  a_{n_1}    & a_{n_1-1}   &\cdots  & a_0 \\
b_{n_2}   & b_{n_2-1}    &\cdots & b_0   \\
         & \ddots    &    \ddots      &  & \ddots     \\
         &           &  b_{n_2}    & b_{n_2-1}   &\cdots  & b_0 \\
\end{vmatrix}.
$$
Denote by $C_1,C_2$ the points in $\bP^s(\overk)$ whose coordinates are the coefficients in $\VX,Y$ of $P_1$ and $P_2$ respectively, where $s$ is the maximum of the numbers of terms of $P_1$ and $P_2$.
By the definitions of determinant, we can write
$$
{\rm res}_Y(P_1,P_2)=\sum_{\bfd} \left(\sum_{j=1}^{\ell_\bfd}\beta_{\bfd,j} \bfm_{\bfd,j}\bfn_{\bfd,j} \right)\VX^{\bfd}
$$
where $\beta_{\bfd,j}, \ell_\bfd\in \Z, \ell_\bfd\geq 0$, $\bfm_{\bfd,j}$ is a monomial in the coordinates of $C_1$ of total degree $n_2$ and $\bfn_{\bfd,j}$ is a monomial in the coordinates of $C_2$ of total degree $n_1$. Viewing $\sum_{j=1}^{\ell_\bfd}\beta_{\bfd,j} \bfm_{\bfd,j}\bfn_{\bfd,j}$ as a polynomial homogeneous in the coordinates of $C_1$ of degree $n_2$ and homogeneous in the coordinates of $C_2$ of degree $n_1$ with coefficients in $\Z$, by Proposition~\ref{prop:heightproperty2}, one has that
\begin{align*}
 \frakh(\res_Y(P_1,P_2))&=\frakh\left(\left(\dots:\sum_{j=1}^{\ell_\bfd}\beta_{\bfd,j} \bfm_{\bfd,j}\bfn_{\bfd,j}:\cdots\right)\right)\\
 &\leq n_2\frakh(C_1)+n_1\frakh(C_2)=n_2\frakh(P_1)+n_1\frakh(P_2).
\end{align*}
\end{proof}
\begin{corollary}
\label{cor:subsitution}
Suppose that $P\in \overk[X,Y]$ and $a,b\in \overk$. Then $$\frakh(P(X+a,Y+b))\leq \frakh(P)+\deg(P,X)\frakh(a)+\deg(P,Y)\frakh(b).$$
\end{corollary}
\begin{proof}
 Denote $d_1=\deg(P,X)$ and $d_2=\deg(P,Y)$ and write
 $$P=\sum_{i=0}^{d_1}\sum_{j=0}^{d_2}c_{i,j}X^i Y^j, \,\,c_{i,j}\in\overk.
 $$
 An easy calculation yields that
\begin{align*}
   P(X+a,Y+b)=\sum_{l_1=0}^{d_1}\sum_{l_2=0}^{d_2}\left(\sum_{i=l_1}^{d_1}\sum_{j=l_2}^{d_2}\binom{i}{l_1}\binom{j}{l_2}c_{i,j}a^{i-l_1}b^{j-l_2}\right)X^{l_1}Y^{l_2}.
\end{align*}
Note that $\sum_{i=l_1}^{d_1}\sum_{j=l_2}^{d_2}\binom{i}{l_1}\binom{j}{l_2}c_{i,j}a^{i-l_1}b^{j-l_2}$ is homogeneous in $c_{i,j}$ of degree 1, homogeneous in $1,a$ of degree $d_1$ and homogeneous in $1,b$ of degree $d_2$ with coefficients in $\Z$. By Proposition~\ref{prop:heightproperty2},
\begin{align*}
  \frakh(P(X+a,Y+b))&=\frakh\left(\left(\dots:\sum_{i=l_1}^{d_1}\sum_{j=l_2}^{d_2}\binom{i}{l_1}\binom{j}{l_2}c_{i,j}a^{i-l_1}b^{j-l_2}:\cdots\right)\right)\\
  &\leq \frakh(P)+d_1\frakh(a)+d_2\frakh(b).
\end{align*}
\end{proof}

When $\frakh$ is an absolutely logarithmic height defined in an algebraic number field, the results in Corollaries~\ref{cor:resultant} and \ref{cor:subsitution} with error terms have already been proved in \cite{Abouzaid}.

\begin{corollary}
 \label{cor:solutionsoflinearsystem} Suppose that $M=(a_{i,j})$ is an $l\times m$ matrix with $a_{i,j}\in \overk$ and $\frakh(a_{i,j})\leq \kappa$. Assume that the linear system $M\bfx=0$ has a nonzero solution. Then $M\bfx=0$ has a nonzero solution $\bfa$ with $\frakh(\bfa)\leq r^2(r+1) \kappa$, where $\bfa$ is viewed as a point in $\bP^{m-1}(\overk)$ and $r=\rank(M)$.
 \end{corollary}
 \begin{proof}
 The assertion is clear in the case $M=0$. Suppose that $M\neq 0$ and $r=\rank(M)$. Without loss of generality, assume that the first $r$-rows of $M$ are linearly independent over $\overk$ and denote by $\tilde{M}$ the $r\times m$ matrix formed by those rows. Then $M$ and $\tilde{M}$ have the same solution space and thus there is no harm to replace $M$ by $\tilde{M}$. Since $\tilde{M}\bfx=0$ has a nonzero solution, $r<m$. Without loss of generality, we further assume that the matrix $M_1$ formed by the first $r$-columns of $\tilde{M}$ is invertible. Denote by $\bfb$ the $(r+1)$-th column of $\tilde{M}$. Using Cramel's rule, $(D_1/\det(M_1),\dots,D_r/\det(M_1))^t$ is the solution of $M_1\bfx=-\bfb$, where $D_i$ is the determinant of the matrix obtained by replacing the $i$-th column of $M_1$ by $-\bfb$, and $(\cdot)^t$ denotes the transpose of a vector. Set
 $$\bfa=(D_1,\dots,D_r,\det(M_1),\underbrace{0,\dots,0}_{m-r-1})^t.$$
 Then $\bfa$ is a solution of $\tilde{M}\bfx=0$. Note that $D_i$ and $\det(M_1)$ are homogeneous in $a_{1,1},\dots,a_{r,r+1}$ of degree $r$. By Proposition~\ref{prop:heightproperty2}, $\frakh(\bfa)\leq r\frakh((a_{1,1}:\dots:a_{r,r+1}))$. By Corollary~\ref{cor:heightproperty12},
 $$\frakh((a_{1,1}:\dots:a_{r,r+1}))\leq \sum_{i,j}\frakh(a_{i,j})\leq r(r+1)\kappa.$$
 So $\frakh(\bfa)\leq r^2(r+1)\kappa$.
 \end{proof}
Note that all valuations constructed by places of $L$ over $k$ are non-archimedean (see page 62 of \cite{lang} for the construction). By Proposition 2.4 on page 57 in \cite{lang} with $s=0$, one sees that if $G$ and $H$ are polynomials in $\overk[X_1,\dots,X_m]$, then
$$\frakh(GH)=\frakh(G)+\frakh(H),$$
from which we have the following proposition.

\begin{prop}\label{prop:heightproperty3}
\begin{enumerate}
\item Suppose that $G,H \in \overk[X_1,\dots,X_m]$ and $G$ divides $H$. Then $\frakh(G)\leq \frakh(H)$.
\item Suppose that $H$ is a nonzero polynomial in $\overk[X]$ and $a$ is a zero of $H$ in $\overk$. Then $\frakh(a)\leq \frakh(H)$.
\end{enumerate}
\end{prop}
\begin{proof}
The first assertion is clear. The second one follows from the facts that $X-a$ divides $H$ and $\frakh(X-a)=\frakh(a)$.
\end{proof}

The following result is claimed on page 13 of \cite{Serre}. We present a proof here for completeness.
\begin{prop}
\label{prop:heightproperty4}
Suppose that $\bfa=(a_0:\dots:a_n)\in \bP^n(\overk)$.
Let $\bfb$ be a point in $\bP^{{{d+n \choose n}-1}}(\overk)$ with all monomials in $a_0,\dots,a_n$ of total degree $d$ as coordinates. Then
$\frakh(\bfb)=d\frakh(\bfa)$.
\end{prop}
\begin{proof}
Due to Proposition~\ref{prop:heightproperty2}, one sees that $\frakh(\bfb)\leq d\frakh(\bfa)$. It remains to prove the converse. Let $L=k(t,a_0,\dots,a_n)$. For each place $\frakp$ of $L$ over $k$, one has that
\begin{align*}
   \max \{-\ord_\frakp(a_0^{s_0}\dots a_n^{s_n}) \mid s_i\geq 0, s_0+\dots+s_n=d\}&\geq \max_{i=0}^n\{-\ord_\frakp(a_i^d)\}\\
   &=d\max_{i=0}^n \{-\ord_\frakp(a_i)\}.
\end{align*}
By definition, $\frakh(\bfb)\geq d\frakh(\bfa)$. So $\frakh(\bfb)=d\frakh(\bfa)$.
\end{proof}
\section{The Riemann-Roch spaces}
\label{sec:Riemann-Roch}
Throughout this section, let $K$ be an algebraically closed field of characteristic zero with an absolute logarithmic height $\frakh$ which satisfies the following conditions:
\begin{enumerate}
\item [(A1)] Propositions~\ref{prop:heightproperty1} and~\ref{prop:heightproperty2} hold for $\frakh$. Consequently, Corollaries~\ref{cor:heightproperty11},\ref{cor:heightproperty12},\ref{cor:resultant}, \ref{cor:subsitution} and \ref{cor:solutionsoflinearsystem} hold for $\frakh$.
\item [(A2)]Propositions~\ref{prop:heightproperty3} and~\ref{prop:heightproperty4} hold for $\frakh$.
\end{enumerate}
\begin{remark}
\label{rem:constant} Under the assumption that Corollary~\ref{cor:heightproperty11} holds for the height $\frakh$, one has that $\frakh(a)=0$ for all $a\in \Q$. To see this, we first have that $\frakh(m)=0$ for all $m\in \Z$. Then for $m_1,m_2\in \Z \setminus \{0\},$
$$
   \frakh\left(m_1/m_2\right)\leq \frakh(m_1)+\frakh(1/m_2)=\frakh(m_1)+\frakh(m_2)=0.
$$
So $\frakh(a)=0$ for all $a\in \Q$.
\end{remark}

Let $L$ be an algebraic function field of one variable over $K$ and $D$ a divisor of $L$. Suppose that $\calL_K(D)\neq \{0\}$. In this section, we are going to give bounds for the degrees and height of a certain nonzero element in $\calL_K(D)$. Let us start with two lemmas.
\begin{lemma}
\label{lm:multiplicationofpowerseries}
Let $f_i=\sum_{s\geq 0} a_{i,s}z^s\in K[[z]]$ for $i=1,\dots,r$. Suppose that $\frakh(a_{i,0})\leq \frakh(a_{i,1})\leq\cdots$ for all $i$. Write $\prod_{i=1}^r f_i=\sum_{s\geq 0} c_s z^{s}$ with $c_s\in K$. Then $$\frakh(c_s)\leq \sum_{i=1}^r (s+1)\frakh(a_{i,s}).$$
\end{lemma}
\begin{proof}
One can easily check that
$$
  c_s=\sum_{\substack{0\leq l_1,\dots,l_r\leq s,\\ l_1+\dots+l_r=s}} a_{1,l_1} \cdots a_{r,l_r}.
$$
By Corollary~\ref{cor:heightproperty11}, one has that $$\frakh(c_s)\leq \sum_{i=1}^r \sum_{j=0}^s\frakh(a_{i,j})\leq \sum_{i=1}^r (s+1)\frakh(a_{i,s}).$$
\end{proof}
\begin{lemma}
\label{lm:formalpowerseriessolutions}
Suppose that $Q\in K[z,Y]$ and $f=\sum_{i\geq 0} a_iz^i \in K[[z]]$ with $Q(z,f)=0$. Then for $i\geq 0$,
$$\frakh(a_i)\leq (\deg(Q,Y)+1)^{i}\frakh(Q).$$
\end{lemma}
\begin{proof}
Denote $n=\deg(Q,Y)$.
Dividing $Q$ by some power of $z$ if necessary, we may assume that $z \nmid Q$. Note that this operation does not change the height of $Q$. It is easy to verify that $\frakh(Q(0,Y))\leq\frakh(Q)$. Since $Q(0,Y)\neq 0$ and $Q(0,a_0)=0$, by Proposition~\ref{prop:heightproperty3}, $\frakh(a_0)\leq \frakh(Q(0,Y))\leq \frakh(Q)$. Now set $Q_1=Q(z,a_0+zY)$.  Then by Corollary~\ref{cor:subsitution},
$$
  \frakh(Q_1)=\frakh(Q(z,a_0+Y))\leq \frakh(Q)+n\frakh(a_0)\leq \frakh(Q)+n\frakh(Q)\leq (n+1)\frakh(Q).
$$
Again, we may assume that $z\nmid Q_1$. One sees that $a_1+a_2z+\cdots$ is a solution of $Q_1(z,Y)=0$. Using a similar argument, one has that $\frakh(a_1)\leq \frakh(Q_1)\leq (n+1)\frakh(Q).$ Set $Q_{i+1}=Q_i(z,a_i+zY)$ for $i=1,2,\dots$. Repeating the previous process yields that $\frakh(a_i)\leq (n+1)^i\frakh(Q)$.
\end{proof}

 Now suppose that $L=K(x,y)$ where $x$ is transcendental over $K$ and $[L:K(x)]<\infty$. Furthermore, assume that $P\in K[X,Y]$ is a nonzero irreducible polynomial such that $P(x,y)=0$. Let us first adapt a result given in \cite{Schmidt:Constructionandestimationofbasesinfunctionfields} on the degree bound for a basis of a Riemann-Roch space. For this, we need to recall some notations introduced in \cite{Schmidt:Constructionandestimationofbasesinfunctionfields}. Write
\begin{equation}
\label{eq:algebraiccurves}
   P(X,Y)=A_0(X)Y^n+A_1(X)Y^{n-1}+\dots+A_n(X),
\end{equation}
where $A_i\in K[X]$, $A_0\neq 0$ and $\deg (P,Y)=n$. Set
$$y_1=1,\,\, y_2=A_0(x)y,\,\,\dots,\,\,y_n=A_0(x)y^{n-1}+\dots+A_{n-2}(x)y.$$
Then the $y_i$'s are integral over $K[x]$. To see this, note that $y_i$ is integral over $K[x]$ if and only if $y_i$ has no pole lying above $x-c$ for any $c\in K$. Suppose that $\frakp$ is a pole of some $y_i$ lying above $x-c$ for some $c\in K$. Then $\frakp$ is a pole of $y$ and thus a zero of $y^{-1}$. On the other hand,
$$y_i=A_0(x)y^{i-1}+\dots+A_{i-2}(x)y=-A_{i-1}(x)-A_{i}(x)y^{-1}-\dots-A_n(x)y^{-n+i-1}$$
which implies that $\frakp$ is not a pole of $y_i$, a contradiction. Let $\bfd(X)$ be the discriminant of $P$ with respect to $Y$. Let $D=\sum_{\frakp} d_{\frakp} \frakp$ be a divisor of $L$. It is clear that $\calL_K(D)=K$ if $D$ is a zero divisor. Assume that $D$ is not a zero divisor in the rest of this article. In what follows, we agree with the following notations.
{\begin{notation}
\label{notation:divisor}
\[
\begin{aligned}
&\delta_D=\sum_{\frakp}|d_\frakp|,\\
&\rho=\tdeg(P),\\[2mm]
& \bfq_D(X)=\bfd(X)^{\rho(\rho+\delta_D)}\prod_{\substack{\frakp \in \supp(D),\\\ord_\frakp(x)\geq 0}}(X-\pi_\frakp(x))^{\rho(\rho+\delta_D)},\\[2mm]
 &U=\supp(\divs(x)^{-})\cup \supp(\divs(y)^{-})\cup \supp(\divs(\bfq_D(x))),\\
 &h(D)=\max\{\frakh(P),\max\{\frakh(\pi_\frakp(x)) \mid \forall\,\frakp\in U\}\}.
\end{aligned}
\]
\end{notation}}
\begin{remark}
\label{rem:convention}
Note that $\supp(D)\subset U$. To see this, for $\frakp\in \supp(D)$ with $\ord_\frakp(x)\geq 0$, $\frakp$ is a zero of $x-\pi_\frakp(x)$ and thus a zero of $\bfq_D(x)$. So $\frakp\in \supp(\divs(\bfq_D(x)))\subset U$. For $\frakp\in \supp(D)$ with $\ord_\frakp(x)<0$, $\frakp\in \supp(\divs(x)^{-})$ which is obvious in $U$.
\end{remark}
In \cite{Schmidt:Constructionandestimationofbasesinfunctionfields}, when $K=\C$, Schmidt gave a degree bound for a basis of the Riemann-Roch space $\calL_{\C}(D)$ of a divisor $D$. Moreover, he proved that if $P$ has coefficients in a subfield $k$ of $\C$ and $D$ is defined over $k$ then there is a basis of $\calL_\C(D)$ whose elements are in $k(x,y)$.
After small modifications of Schmidt's results, we are able to prove that Schmidts's result on degree bound is also valid for the algebraically closed field $K$. Suppose that $k\subset K$ is an algebraically closed subfield such that $P(X,Y)\in k[X,Y]$. By Definition~\ref{def:divisors}, one sees that each divisor of $k(x,y)$ is also a divisor of $K(x,y)$. Suppose that $f\in k(x,y)\setminus\{0\}$. To avoid confusion, we denote by $\divs_k(f)$ the divisor of $f$ viewed as an element in $k(x,y)$ and by $\divs_K(f)$ the divisor of $f$ viewed as an element in $K(x,y)$.
\begin{lemma}
\label{lm:divisors}
 Suppose that $k\subset K$ is an algebraically closed subfield such that $P(X,Y)\in k[X,Y]$. Then for every $f\in k(x,y)\setminus\{0\}$, $\divs_k(f)=\divs_K(f)$.
\end{lemma}
\begin{proof}
If $f\in k$ then there is nothing to prove. Suppose that $f\notin k$. Since a zero (resp. pole) of $f$ in $k(x,y)$ is also a zero (resp. pole) of $f$ in $K(x,y)$, $f\notin K$. Due to Theorem 4 on page 18 of \cite{chevalley},
$$\deg(\divs_k(f)^{+})=\deg(\divs_k(f)^{-})=[k(x,y):k(f)].$$
Similarly, one has that
$$\deg(\divs_K(f)^{+})=\deg(\divs_K(f)^{-})=[K(x,y):K(f)].$$
As $k$ is algebraically closed, $[k(x,y):k(f)]=[K(x,y):K(f)]$. This implies that
$\deg(\divs_k(f)^{+})=\deg(\divs_K(f)^{+})$ and $\deg(\divs_k(f)^{-})=\deg(\divs_K(f)^{-})$. On the other hand, one has that $\divs_k(f)^{+}, \divs_k(f)^{-}$ are divisors of $K(x,y)$ and $\divs_k(f)^{+}\leq \divs_K(f)^{+}, \divs_k(f)^{-}\leq \divs_K(f)^{-}$. Therefore $\divs_k(f)^{+}= \divs_K(f)^{+}$ and $\divs_k(f)^{-}=\divs_K(f)^{-}$. Consequently, $\divs_k(f)=\divs_K(f)$.
\end{proof}
\begin{prop}
\label{prop:degreebounds}
Let $\rho,\delta_D,\bfq_D$ be as in Notation~\ref{notation:divisor}. Then there are integers $\pi_1,\dots,\pi_n$, and a monic factor $q$ of $\bfq_D$ with $\deg(q)<\rho(\rho+\delta_D)$, $B_{i,j}\in K[X]$ with $\deg(B_{i,j})<2\rho(\rho+2\delta_D)$ such that $\calL_K(D)$ has a basis of the type
\begin{equation}
\label{eq:types}
   x^l\left(\sum_{j=1}^{n} \frac{B_{i,j}(x)}{q(x)}y_j\right)
\end{equation}
where $i$ runs over all integers $s\in \{1,2,\dots,n\}$ satisfying $\pi_s\geq 0$ and $l$ runs over all integers in $\{0,1,\dots,\pi_i\}$.
\end{prop}
\begin{proof}
If $\calL_K(D)=\{0\}$, we take all $\pi_i< 0$ for $i=1,\dots,n$ and the assertion is obvious. Now assume $\calL_K(D)\neq \{0\}$ and $a_1,\dots,a_m\in L$ is a basis of $\calL_K(D)$ over $K$.
Let $k\subset K$ be a field finitely generated over $\Q$ such that $P\in k[X,Y],\,a_1,\dots,a_m\in k(x,y)$ and the center of $\frakp$ has coordinates in  $k\cup\{\infty\}$ for every place $\frakp$ in $\supp(D)$. Then the irreducible parametrization corresponding to $\frakp$ has coordinates in $\bar{k}((z))$ {for any $\frakp \in \supp(D)$, where $\bar{k}$ is the algebraic closure of $k$}. We embed $\bar{k}$ into $\C$ and view $P$ as a polynomial in $\C[X,Y]$. Then $P$ is irreducible over $\C$ because $P$ is irreducible over $\bar{k}$. Denote by $\tilde{L}$ the field of fractions of $\C[X,Y]/(P)$, where $(P)$ stands for the ideal in $\C[X,Y]$ generated by $P$. Then $\bar{k}(x,y)$ can be viewed as a subfield of $\tilde{L}$. Note that $D$ is still a divisor of both $\bar{k}(x,y)$ and $\tilde{L}$. By Lemma~\ref{lm:divisors}, $\calL_{\bar{k}}(D)\subset \calL_{\C}(D)$. Since $D$ is defined over $\bar{k}$, by Theorems A2 and B2 of \cite{Schmidt:Constructionandestimationofbasesinfunctionfields}, $\calL_{\C}(D)$ has a basis of the type (\ref{eq:types}) with $B_{i,j}\in \bar{k}[X]$, $\deg(q)\leq \deg(P,Y)\delta_D+\deg(\bfd)/2$ and $\deg(B_{i,j})\leq \deg(P,Y)(\deg(P,X)+3\delta_D)+\deg(q)$. Note that $\deg(\bfd)\leq (2\deg(P,Y)-1)\deg(P,X)<2\rho^2$. One sees that $\deg(q)<\rho(\rho+\delta_D)$ and $\deg(B_{i,j})<2\rho(\rho+2\delta_D)$. Due to Theorem 1 on page 90 of \cite{chevalley}, the vector spaces $\calL_{\bar{k}}(D)$ and $\calL_{\C}(D)$ have the same dimension and then $\calL_{\bar{k}}(D)$ has a basis of the type (\ref{eq:types}) with $B_{i,j}\in \bar{k}[X]$. Since $\calL_{\bar{k}}(D)\subset \calL_K(D)$ by Lemma~\ref{lm:divisors} and $a_1,\dots,a_m\in k(x,y)$, $\calL_{\bar{k}}(D)$ and $\calL_K(D)$ have the same dimension by Theorem 1 on page 90 of \cite{chevalley}. These imply that $\calL_K(D)$ has a basis of the type (\ref{eq:types}) with $B_{i,j}\in \bar{k}[X]\subset K[X]$.
 \end{proof}
 \begin{corollary}
 \label{cor:Riemann-Roch spaces}
 Let $\rho,\delta_D, q$ be as in Proposition~\ref{prop:degreebounds} and $\tilde{D}=D-\divs(q(x))$. Suppose that $\calL_K(\tilde{D})\neq \{0\}$. Then $\calL_K(\tilde{D})$ contains {a nonzero} element of the type
 $$
    \sum_{j=0}^{n-1} \tilde{B}_j(x)y^j
 $$
 where $\tilde{B}_j\in K[X]$ with $\deg(\tilde{B}_j)< 4\rho(\rho+\delta_D)$.
 \end{corollary}
 \begin{proof}
By Proposition~\ref{prop:degreebounds}, there are integers $\pi_1,\dots,\pi_n$, and $B_{i,j}\in K[X]$ with $\deg(B_{i,j})<2\rho(\rho+2\delta_D)$ such that $\calL_K(D)$ has a basis of the type
\begin{equation}
\label{eq:basisofriemann-roch}
   x^l\left(\sum_{j=1}^{n} \frac{B_{i,j}(x)}{q(x)}y_j\right)
\end{equation}
where $i$ runs over all integers $s\in \{1,2,\dots,n\}$ satisfying $\pi_s\geq 0$ and $l$ runs over all integers in $\{0,1,\dots,\pi_i\}$. Note that $\calL_K(D)=\frac{1}{q(x)}\calL_K(\tilde{D})$. Thus $\calL_K(D)\neq \{0\}$ which implies that not all $\pi_i$ are negative. Suppose that $\pi_{i_0}\geq 0$. Setting $l=0$ in (\ref{eq:basisofriemann-roch}) yields that $\sum_{j=1}^{n} \frac{B_{i_0,j}(x)}{q(x)}y_j$ is an element in $\calL_K(D)$. Write $\sum_{j=1}^n B_{i_0,j}(x)y_j=\sum_{j=0}^{n-1}\tilde{B}_j(x) y^j$ where $\tilde{B}_j\in K[X]$. Note that $y_1=1$ and $y_j=\sum_{s=1}^{j-1}A_{j-1-s}(x)y^s$ for $j>1$. One has that
{\begin{align*}
   \sum_{s=0}^{n-1}\tilde{B}_s(x)y^s&=B_{i_0,1}(x)+\sum_{j=2}^n\sum_{s=1}^{j-1} A_{j-1-s}(x)B_{i_0,j}(x)y^s\\
   &=B_{i_0,1}(x)+\sum_{s=1}^{n-1}\left(\sum_{j=s+1}^n A_{j-1-s}(x)B_{i_0,j}(x)\right)y^s.
\end{align*}
Therefore $\tilde{B}_0(x)=B_{i_0,1}(x)$ and $\tilde{B}_s(x)=\sum_{j=s+1}^n A_{j-1-s}(x)B_{i_0,j}(x)$ for $s\geq 1$ and so}
$$
   \deg(\tilde{B}_s)\leq \max_{j}\deg(B_{i_0,j})+\max_{j}\deg(A_j)\leq 2\rho(\rho+2\delta_D)+\rho<4\rho(\rho+\delta_D).
$$
The corollary then follows from the fact that $\calL_K(D)=\frac{1}{q(x)}\calL_K(\tilde{D})$.
 \end{proof}
 In the rest of this section, let us estimate the heights of the coefficients of $\tilde{B}_j$. We first estimate the heights of the coefficients of a place represented by an irreducible parametrization of $P(X,Y)=0$.
 \begin{prop}
\label{prop:heightofplaces}
Let $\rho=\tdeg(P)$.
Suppose that $(z^\mu+a,z^{\nu}(c_0+c_{\ell_1}z^{\ell_1}+\cdots))$ is a place of $P(X,Y)=0$. Then
$$\frakh(c_i)\leq (\rho+1)^{i+1}\max\{\frakh(P),\frakh(a)\}$$
where $c_i=0$ if $i\neq \ell_j$ for all $j\geq 1$ and $i\neq 0$.
\end{prop}
\begin{proof}
We first consider the case $\mu>0$.
Set $\bar{P}(z,Y)=z^dP(z^\mu+a,z^{\nu}Y)$ where $d$ is the integer such that $\bar{P}\in K[z,Y]$ and $z\nmid \bar{P}$. By Corollary~\ref{cor:subsitution}, one can verity that
$$\frakh(\bar{P})\leq\frakh(P)+\deg(P,X)\frakh(a)\leq(\rho+1)\max\{\frakh(P),\frakh(a)\}.$$
As $c_0+c_{\ell_1}z^{\ell_1}+\cdots$ is a solution of $\bar{P}(z,Y)=0$ and $\deg(P,Y)=\deg(\bar{P},Y)$, by Lemma~\ref{lm:formalpowerseriessolutions}, one sees that
$$\frakh(c_i)\leq (\deg(\bar{P},Y)+1)^i\frakh(\bar{P})\leq (\rho+1)^{i+1}\max\{\frakh(P),\frakh(a)\}.$$
Suppose that $\mu<0$. Similarly, set $\bar{P}(z,Y)=z^dP(z^\mu,z^{\nu}Y)$ where $d$ is the integer such that $\bar{P}\in K[z,Y]$ and $z\nmid \bar{P}$. Then $\frakh(P)=\frakh(\bar{P})$ and $\deg(P,Y)=\deg(\bar{P},Y)$. Since $c_0+c_{\ell_1}z^{\ell_1}+\cdots$ is a solution of $\bar{P}(z,Y)=0$, by Lemma~\ref{lm:formalpowerseriessolutions}, $\frakh(c_i)\leq (\rho+1)^{i+1}\frakh(P)$.
\end{proof}
{For a place $\frakp=(z^\mu+a,z^{\nu}(c_0+c_1z+{\cdots}))$ of $P(X,Y)=0$, the series $(z^\mu+a)^l (z^{\nu}(c_0+c_1z+{\cdots}))^j$ is called the expansion of $x^l y^j$ at $\frakp$, denoted by $x^l y^j|_{(x,y)=\frakp}$ for brevity.}
 \begin{lemma}
 \label{lm:heightsofcoefficients2}
Let $\rho=\tdeg(P)$. For $l\geq 0, j\in \{0,\dots,n-1\}$ and a place $\frakp$, $x^l y^j$ has an expansion at $\frakp$ of the type
 $$
    \left.x^l y^j\right|_{(x,y)=\frakp}=z^{d_{\frakp,l,j}}\sum_{s=0}^\infty \beta_{\frakp,l,j,s} z^s
 $$
 where $d_{\frakp,l,j}$ is an integer greater than $-l\rho-\rho^2$ and $\beta_{\frakp,l,j,s}\in K$ with
 $$\frakh(\beta_{\frakp,l,j,s})\leq ((s+1)^2 (\rho+1)^{s+2}+l)\max\{\frakh(P),\frakh(\pi_\frakp(x))\}.$$
 \end{lemma}
 \begin{proof}
 Suppose that $$\frakp=(\bfx(z),\bfy(z))=(z^\mu+a,z^\nu(c_0+c_{\ell_1}z^{\ell_1}
 +\cdots)).$$
 Then $\frakh(\pi_\frakp(x))=\frakh(a)$. To see this, if $\mu>0$ then $\pi_\frakp(x)=a$ and we are done, if $\mu<0$ then $\pi_\frakp(x)=\infty$ and $a=0$ and thus $\frakh(\pi_\frakp(x))=0=\frakh(a)$.
By Proposition~\ref{prop:heightofplaces}, for $i\geq 0$,
 $$\frakh(c_i)\leq (\rho+1)^{i+1}\max\{\frakh(P),\frakh(a)\}$$
 where $c_i=0$ if $i\neq l_j$ for all $j\geq 1$ and $i\neq 0$.
 Write $\bfy(z)^j=z^{j\nu}\sum_{s\geq 0} b_{j,s}z^s$. By Lemma~\ref{lm:multiplicationofpowerseries} with $f_i=\sum_{s\geq 0} c_sz^s$, one sees that
 \begin{align*}
   \frakh(b_{j,s})\leq \sum_{i=1}^j (s+1)\frakh(c_s)\leq j(s+1)\frakh(c_s)\leq (s+1)(\rho+1)^{s+2}\max\{\frakh(P),\frakh(a)\}.
 \end{align*}
The last inequality holds because $j\leq n-1<\rho+1$.

We first consider the case $\mu>0$.
Note that $(z^\mu+a)^l=\sum_{s=0}^l {l \choose s} a^{l-s}z^{s\mu}$. This implies that $(z^\mu+a)^l\bfy(z)^j$ has an expansion of the type $z^{e_{l,j}} \sum_{s\geq 0} \beta_{\frakp,l,j,s} z^s$ at $z=0$, where $e_{l,j}=j\nu$ and
$$
  \beta_{\frakp,l,j,s}=\sum_{i=0}^{l} b_{j,s-i\mu}{l \choose i}a^{l-i}
$$
with $b_{j,i}=0$ if $i<0$.
Therefore by Corollary~\ref{cor:heightproperty11},
\begin{align*}
\frakh(\beta_{\frakp,l,j,s})&\leq \sum_{i=0}^{s}\frakh(b_{j,i})+l\frakh(a)\leq ((s+1)^2(\rho+1)^{s+2}+l)\max\{\frakh(P),\frakh(a)\}.
\end{align*}
 Set $d_{\frakp,l,j}=j\nu$. Then we has the required expansion for $x^ly^j$ at $\frakp$. Finally, as $|\nu|\leq |\ord_\frakp(y)|\leq \rho$, one has that $d_{\frakp,l,j}>-\rho^2\geq-\rho l-\rho^2.$

 Now suppose that $\mu<0$. In this case, one easily sees that $d_{\frakp,l,j}=j\nu+l\mu$ and
$\beta_{\frakp,l,j,s}=b_{j,s}$. As $|\mu|\leq |\ord_\frakp(x)|\leq \rho$, one has that $d_{\frakp,l,j}>-l\rho-\rho^2$.
 \end{proof}

 Let $\bfc=(\dots,c_{l,j},\dots)$ be a vector with indeterminate coordinates and set $$g(\bfc)=\sum_{j=0}^{n-1}\sum_{l=0}^{4\rho(\rho+\delta_D)-1} c_{l,j}x^l y^j.$$
 \begin{prop}
 \label{prop:heightsofRiemann-Roch}
 Let $\rho,\delta_D,h(D)$ be as in Notation~\ref{notation:divisor}. Let $\tilde{D}$ be as in Corollary~\ref{cor:Riemann-Roch spaces}. Suppose that $\calL_K(\tilde{D})\neq \{0\}$. Then $\calL_K(\tilde{D})$ contains a nonzero element of the type
 \begin{equation}
 \label{eq:riemann-rochelement}
     g(\bfa)=\sum_{j=0}^{n-1} \sum_{l=0}^{4\rho(\rho+\delta_D)-1} a_{l,j}x^l y^j
 \end{equation}
 with $$\frakh(\bfa)\leq 1600(\rho+\delta_D)^6(\rho+1)^{5(\rho+\delta_D)^3-11}h(D),$$
 where $a_{l,j}\in K$, at least one of $a_{l,j}$ equals 1 and $\bfa$ is viewed as a projective point.
 \end{prop}
\begin{proof}
Let $U$ be as in Notation~\ref{notation:divisor}.
By Lemma~\ref{lm:heightsofcoefficients2}, for each place $\frakp$, $j=0,\dots,n-1$ and $l\geq 0$, $x^ly^j$ has an expansion at $\frakp$ of the type
$$
    \left.x^l y^j\right|_{(x,y)=\frakp}=z^{d_{\frakp,l,j}}\sum_{s=0}^\infty \beta_{\frakp,l,j,s} z^s
 $$
 where $d_{\frakp,l,j}$ is an integer greater than $-l\rho-\rho^2$ and $\beta_{\frakp,l,j,s}\in K$ with
 $$\frakh(\beta_{\frakp,l,j,s})\leq ((s+1)^2(\rho+1)^{s+2}+l)\max\{\frakh(P),\frakh(\pi_\frakp(x))\}.$$
Set $o=\min_{\frakp,l,j}\{d_{\frakp,l,j}\}$ and write
$$
   x^l y^j|_{(x,y)=\frakp}=z^o \sum_{s=0}^\infty \alpha_{\frakp,l,j,s} z^s
$$
One can easily see that $\alpha_{\frakp,l,j,s}=0$ if $s<d_{\frakp,l,j}-o$ and $\alpha_{\frakp,l,j,s}=\beta_{\frakp,l,j,s+o-d_{\frakp,l,j}}$ if $s\geq d_{\frakp,l,j}-o$. Therefore for $s\geq 0$,
\begin{align*}
  \frakh(\alpha_{\frakp,l,j,s})\leq \frakh(\beta_{\frakp,l,j,s})\leq ((s+1)^2(\rho+1)^{s+2}+l)\max\{\frakh(P),\frakh(\pi_\frakp(x))\}.
\end{align*}
Then for each place $\frakp$, $g(\bfc)$ has an expansion at $\frakp$ of the type
$$
  g(\bfc)|_{(x,y)=\frakp}=z^o\sum_{s\geq 0} \left(\sum_{l,j} c_{l,j}\alpha_{\frakp,l,j,s}\right) z^s.
$$
Suppose that $\bar{\bfc}=(\dots:\bar{c}_{l,j}: \cdots)$ where $\bar{c}_{l,j}\in K$. Note that a pole of $g(\bar{\bfc})$ is either a pole of $x$ or a pole of $y$ and so all poles of $g(\bar{\bfc})$ are in $U$. Write $\tilde{D}=\sum_{\frakp} m_\frakp \frakp$. Then $g(\bar{\bfc})\in \calL_K(\tilde{D})$ if and only if $\ord_\frakp(g(\bar{\bfc}))\geq -m_\frakp$ for every $\frakp\in \supp(\tilde{D})$ and $\ord_\frakp(g(\bar{\bfc}))\geq 0$ for every $\frakp\in U \setminus \supp(\tilde{D})$, i.e. $\ord_z(g(\bar{\bfc})|_{(x,y)=\frakp})\geq -m_\frakp$ for every $\frakp\in \supp(\tilde{D})$ and $\ord_z(g(\bar{\bfc})|_{(x,y)=\frakp})\geq 0$ for every $\frakp\in U \setminus \supp(\tilde{D})$. Equivalently, $g(\bar{\bfc})\in \calL_K(\tilde{D})$ if and only if $\bar{\bfc}$ is a solution of the following linear system
\begin{equation}
\label{eq:system}
\begin{aligned}
  &\bigcup_{\frakp\in \supp(\tilde{D})}\left\{\sum_{l,j} c_{l,j}\alpha_{\frakp,l,j,s}=0\mid s=0,\dots,-m_\frakp-o-1\right\}\bigcup\\
  &\bigcup_{\frakp\in U\setminus\supp(\tilde{D})}\left\{\sum_{l,j} c_{l,j}\alpha_{\frakp,l,j,s}=0 \mid s=0,\dots,-o-1\right\}.
\end{aligned}
\end{equation}
Note that $\tilde{D}=D-\divs(q(x))$. By Remark~\ref{rem:convention}, $\supp(D)\subset U$ and thus $\supp(\tilde{D})\subset U$. By definition, $\frakh(\pi_\frakp(x))\leq h(D)$ for all $\frakp\in U$. So for $l\leq 4\rho(\rho+\delta_D)-1$ and $\frakp\in U$,
\begin{align*}
   \frakh(\alpha_{\frakp,l,j,s})&\leq ((s+1)^2(\rho+1)^{s+2}+l)h(D)\\
   &\leq (\rho+\delta_D)(s+1)^2(\rho+1)^{s+3}h(D).
\end{align*}
The second inequality holds because $(s+1)^2(\rho+1)^{s+3}-4\rho \geq (s+1)^2(\rho+1)^{s+2}$.
In what follows, we shall estimate $-m_\frakp$ when $m_\frakp<0$.
Note that
$$\deg(\divs(q(x))^{+})=\deg(q)\deg(\divs(x)^{+})=\deg(q)n\leq \rho^2(\rho+\delta_D).$$
Hence $|m_\frakp|\leq \delta_D+\deg(\divs(q(x))^{+})<(\rho+1)^2(\rho+\delta_D)$. Since $o>-\rho l-\rho^2$ and $l\leq 4\rho(\rho+\delta_D)-1$, one has that
\begin{align*}
-m_\frakp-o-1&< (\rho+1)^2(\rho+\delta_D)+\rho l+\rho^2\\
&\leq 5\rho^2(\rho+\delta_D)+2\rho(\rho+\delta_D)+\rho^2+\delta_D\\
&=5(\rho+1)^2(\rho+\delta_D)-7\rho^2-8\rho\delta_D-5\rho-4\delta_D\\
&\leq 5(\rho+1)^2(\rho+\delta_D)-24.
\end{align*}
Therefore the heights of the coefficients of the system (\ref{eq:system}) are not greater than
\begin{align*}
   T&\triangleq (\rho+\delta_D)(5(\rho+1)^2(\rho+\delta_D)-23)^2(\rho+1)^{5(\rho+1)^2(\rho+\delta_D)-21}h(D)\\
     &\leq 25(\rho+\delta_D)^3(\rho+1)^{5(\rho+\delta_D)^3-17}h(D).
\end{align*}
The system (\ref{eq:system}) contains $4n\rho(\rho+\delta_D)\leq 4\rho^2(\rho+\delta_D)$ variables and thus the rank of the system (\ref{eq:system}) is not greater than $4\rho^2(\rho+\delta_D)$. By
Corollary~\ref{cor:solutionsoflinearsystem}, the system (\ref{eq:system}) has a nonzero solution $\bar{\bfc}$ with
$$
   \frakh(\bar{\bfc})\leq (4\rho^2(\rho+\delta_D))^2(4\rho^2(\rho+\delta_D)+1)T\leq 1600(\rho+\delta_D)^6(\rho+1)^{5(\rho+\delta_D)^3-11}h(D).
$$
Let $\lambda$ be a nonzero coordinate of $\bar{\bfc}$ and
set $\bfa=\bar{\bfc}/\lambda$. Then $g(\bfa)$ is the desired element.
\end{proof}
\begin{prop}
 \label{prop:equality}
 Let $\bfa$ be as in Proposition~\ref{prop:heightsofRiemann-Roch} and let $\rho,\delta_D, h(D)$ be as in Notation~\ref{notation:divisor}.
 Suppose that $Q_1\in K[X,Z], Q_2\in K[Y,Z]$ are nonzero irreducible polynomials such that $Q_1(x,g(\bfa)/q(x))=0$ and $Q_2(y,g(\bfa)/q(x))=0$. Then
  $$\frakh(Q_1),\frakh(Q_2)\leq 1600(\rho+\delta_D)^6(\rho+1)^{5(\rho+\delta_D)^3-9}h(D).$$
 \end{prop}
 \begin{proof}
 Suppose that $\bfa=(\dots,a_{l,j},\dots)$. Set
 $$G(X,Y,Z)=q(X)Z-\sum_{j=0}^{n-1}\sum_{l=0}^{4\rho(\rho+\delta_D)-1} a_{l,j}X^l Y^j.$$
 Then $\deg(G,X)\leq 4\rho(\rho+\delta_D)-1$ and $g(\bfa)/q(x)$ is a solution of $G(x,y,Z)=0$. Denote $R=\res_Y(P,G)$.
 Since $P(X,Y)$ is irreducible and it does not divide $G(X,Y,Z)$, $R$ is nonzero. Furthermore $R(x,g(\bfa)/q(x))=0$ and then $Q_1$ divides $R$.
 Now let us estimate the height of $q(X)$. Suppose that $b_1,\dots,b_d$ are all roots of $q(X)=0$ where $d=\deg(q)$.
 Then $b_i$ is either a zero of $\bfd(X)$ or $\pi_\frakp(x)$ for some $\frakp\in \supp(D)$. In the first case, $\frakh(b_i)\leq (2\rho-1)\frakh(P)$ by Corollary~\ref{cor:resultant} and in the second case $\frakh(b_i)\leq h(D)$. Therefore $\frakh(b_i)\leq (2\rho-1) h(D)$.
 Each coefficient of $q(X)$ is a homogeneous polynomial in $1,b_1,\dots,b_d$ of degree $d$. By Proposition~\ref{prop:heightproperty2}, $\frakh(q(X))\leq d\frakh((1:b_1:\dots:b_d))$. Due to Corollary~\ref{cor:heightproperty12}, $\frakh((1:b_1:\dots:b_d))\leq \sum_{i=1}^d \frakh(b_i)$. Since $d=\deg(q)\leq \rho(\rho+\delta_D)$,
 $$
    \frakh(q(X))\leq d\sum_{i=1}^d \frakh(b_i)\leq d^2(2\rho-1)h(D)\leq (\rho+\delta_D)^2(2\rho^3-\rho^2)h(D).
 $$
 Let $\bfc$ be a point in some projective space with the coefficients of $q(X)$ and all $a_{l,j}$ as coordinates. By Corollary~\ref{cor:heightproperty12},
 \begin{align*}
   \frakh(\bfc)&\leq \frakh(q)+\frakh(\bfa)\\
   &\leq (\rho+\delta_D)^2(2\rho^3-\rho^2)h(D)+1600(\rho+\delta_D)^6(\rho+1)^{5(\rho+\delta_D)^3-11}h(D)\\
   &\leq 1600(\rho+\delta_D)^6(\rho+1)^{5(\rho+\delta_D)^3-10}h(D).
 \end{align*}
 Equivalently, $\frakh(G)\leq 1600(\rho+\delta_D)^6(\rho+1)^{5(\rho+\delta_D)^3-10}h(D)$.
 Due to Proposition~\ref{prop:heightproperty3} and Corollary~\ref{cor:resultant}, one has that
 \begin{align*}
   \frakh(Q_1)\leq \frakh(R)&\leq \deg(G,Y)\frakh(P)+\deg(P,Y)\frakh(G)\\
   &\leq (\rho-1)h(D)+1600\rho(\rho+\delta_D)^6(\rho+1)^{5(\rho+\delta_D)^3-10}h(D)\\
   &\leq 1600(\rho+\delta_D)^6(\rho+1)^{5(\rho+\delta_D)^3-9}h(D).
 \end{align*}
  Using a similar argument, one has that
 $$\frakh(Q_2)\leq \deg(G,X)\frakh(P)+\deg(P,X)\frakh(G)$$
  which is also less than $1600(\rho+\delta_D)^6(\rho+1)^{5(\rho+\delta_D)^3-9}h(D)$.
 \end{proof}

\section{Main result}
\label{sec:heights}
Throughout this section, let $K, \frakh$ be as in Section~\ref{sec:Riemann-Roch} and $L$ stands for an algebraic function field of one variable over $K$, i.e. $L=K(x,y)$ where $x$ is transcendental over $K$ and $[L:K(x)]<\infty$. Suppose that $\frakp$ is a place of $L$ over $K$. Let $\pi_\frakp$ be defined as in Section~\ref{sec:algebraicfunctionfield}. We start with a height inequality for points on an algebraic curve of special type. This inequality is an easy corollary of Proposition on page 14 of \cite{Serre}. For completeness, we present a detailed proof and estimate the constant term.
\begin{prop}
\label{prop:specialcase}
Suppose that $Q$ is a nonzero polynomial in $K[X,Y]$ satisfying $\deg(Q,Y)=\tdeg(Q)$. Then for each $(\alpha,\beta)\in K^2$ with $Q(\alpha,\beta)=0$,
$$
  \frakh(\beta)\leq \frakh(\alpha)+\frakh(Q).
$$
\end{prop}
\begin{proof}
Suppose that $n=\tdeg(Q)$.
Let $\bfm_0,\dots,\bfm_\ell$ be all monomials in $X,Y$ of total degrees not greater than $n$.  Without loss of generality, we assume that $\bfm_0=1$ and $\bfm_\ell=Y^n$. Write $Q=cY^n+\sum_{i=0}^{\ell-1} b_i \bfm_i$ where $c,b_i\in K$ and $c\neq 0$. In Proposition~\ref{prop:heightproperty2}, we take $r=1, s_1=\ell-1, \Phi=(X_{1,0}:\dots:X_{1,\ell-1}:-\frac{1}{c}\sum_{i=0}^{\ell-1} b_i X_{1,i})$ and $\bfa=(1:\bfm_1(\alpha,\beta):\dots:\bfm_{\ell-1}(\alpha,\beta))$. Then $\Phi(\bfa)=(\bfm_0(\alpha,\beta):\dots:\bfm_\ell(\alpha,\beta))$ and by Propositions~\ref{prop:heightproperty4} and~\ref{prop:heightproperty2}, one has that
\begin{equation}
\label{eq:inequality1}
  n\frakh((1:\alpha:\beta))=\frakh(\Phi(\bfa))\leq \frakh(\bfa)+\frakh(\Phi)=\frakh(\bfa)+\frakh(Q).
\end{equation}
Let $\bfn_0,\dots,\bfn_m$ be all monomials in $X,Y$ of total degrees not greater than $n-1$. One can check that there are $\bfn_{d_1},\dots,\bfn_{d_{n}}$ such that $X\bfn_{d_i}\neq \bfn_j$ for any $i,j$ and
$$
   \{\bfm_i \mid i=0,\dots,\ell-1\}=\{\bfn_i \mid i=0,\dots,m\}\cup \{X\bfn_{d_i} \mid i=1,\dots,n \},
$$
where $\ell=m+n+1$.
In Proposition~\ref{prop:heightproperty2}, we take $r=2, s_1=1, s_2=m$, $\Phi=(\phi_0:\dots:\phi_{\ell-1})$ with $\phi_i=X_{1,0}X_{2,i}$ for $i=0,\dots,m$ and $\phi_i=X_{1,1}X_{2,d_{i-m}}$ for $i=m+1,\dots,\ell-1$, $\bfa_1=(1:\alpha)$ and $\bfa_2=(\bfn_0(\alpha,\beta):\dots:\bfn_m(\alpha,\beta))$. Reordering the subscripts if necessary, we may assume that
$$
  (\bfm_0,\dots,\bfm_{m+n})=(\bfn_0,\dots,\bfn_m,X\bfn_{d_1},\dots,X\bfn_{d_n}).
$$
We then have that $\Phi(\bfa_1,\bfa_2)=\bfa$ and
\begin{equation}
\label{eq:inequality2}
\frakh(\bfa)=\frakh(\Phi(\bfa_1,\bfa_2))\leq \frakh(\bfa_1)+\frakh(\bfa_2)=\frakh(\alpha)+\frakh(\bfa_2).
\end{equation}
By Proposition~\ref{prop:heightproperty4} again, $\frakh(\bfa_2)=(n-1)\frakh((1:\alpha:\beta))$. This together with (\ref{eq:inequality1}) and (\ref{eq:inequality2}) yields that
$$
   \frakh((1:\alpha:\beta))\leq \frakh(\alpha)+\frakh(Q).
$$
The proposition then follows from the fact that $\frakh(\beta)\leq \frakh((1:\alpha:\beta))$.
\end{proof}
As a corollary, we have the following quasi-equivalence of heights for points on an algebraic curve of special type.
\begin{corollary}
\label{cor:simplecase}
Suppose that $Q=\sum_{i=0}^{m}\sum_{j=0}^n a_{i,j}X^iY^j$ with $a_{i,j}\in K$, $m=\deg(Q,X)$ and $n=\deg(Q,Y)$. Assume that for all $0\leq i\leq m$ and $0\leq j\leq n$ if $a_{i,j}\neq 0$ then $mj+ni\leq mn$. Then for each $(\alpha,\beta)\in K^2$ with $Q(\alpha,\beta)=0$,
$$
  n\frakh(\beta)-mn\frakh(Q)\leq m\frakh(\alpha)\leq n\frakh(\beta)+mn\frakh(Q).
$$
\end{corollary}
\begin{proof}
Set $\tilde{Q}=\sum_{i=0}^{m}\sum_{j=0}^n a_{i,j}X^{ni}Y^{mj}$. Then $\deg(\tilde{Q},Y)=\deg(\tilde{Q},X)=\tdeg(\tilde{Q})$ and $\frakh(\tilde{Q})=\frakh(Q)$. Suppose that $(\alpha,\beta)\in K^2$ satisfies $Q(\alpha,\beta)=0$. Then $\tilde{Q}(\alpha^{1/n},\beta^{1/m})=0$. By Proposition~\ref{prop:specialcase}, $\frakh(\beta^{1/m})\leq \frakh(\alpha^{1/n})+\frakh(\tilde{Q})$. By Proposition~\ref{prop:heightproperty1}, one has that $n\frakh(\beta)\leq m\frakh(\alpha)+mn\frakh(Q)$. Similarly, one has that $m\frakh(\alpha)\leq n\frakh(\beta)+mn\frakh(Q)$.
\end{proof}

The polynomial $Q$ usually does not satisfy the assumption of Proposition~\ref{prop:specialcase}, i.e. $\deg(Q,Y)=\tdeg(Q)$. In order to apply Proposition~\ref{prop:specialcase}, Eremenko proved in \cite{eremenko} that if $\divs(y)^{-}\leq \divs(x)^{-}$ then the irreducible polynomial $Q$ with $Q(x,y)=0$ satisfies $\deg(Q,Y)=\tdeg(Q)$. The following lemma is a generalization of Lemma 1 in \cite{eremenko}.
\begin{lemma}
\label{lm:pointinequality}
Assume  that $Q\in K[X,Y]$ is a nonzero polynomial irreducible over $K$ and $\alpha,\beta\in L\setminus K$ satisfying $Q(\alpha,\beta)=0$. Suppose that $$m_1\divs\left(\tau_1(\alpha)\right)^{-} \leq m_2\divs\left(\tau_2(\beta)\right)^{-}$$
where $m_1,m_2$ are positive integers and $\tau_1,\tau_2$ are two linear fractional transformations with coefficients in $\Q$.
Then for every place $\frakp$ of $L$ over $K$,
$$
m_1\frakh(\pi_\frakp(\alpha))\leq m_2\frakh(\pi_\frakp(\beta))+m_1m_2\frakh(Q).
$$
\end{lemma}
\begin{proof}
Write $\tau_1(X)=\frac{a_1 X+b_1}{c_1 X+d_1},\,\tau_2(Y)=\frac{a_2 Y+b_2}{c_2 Y+d_2}$ with $a_i,b_i,c_i,d_i\in \Q$ and $a_id_i-b_ic_i\neq 0$.
Denote $\bar{\alpha}=\tau_1(\alpha)^{m_1}$ and $\bar{\beta}=\tau_2(\beta)^{m_2}$.
Let $\bar{Q}\in K[Z_1,Z_2]$ be a nonzero irreducible polynomial such that $\bar{Q}(\bar{\alpha},\bar{\beta})=0$.
Set
\begin{align*}
&H_1=(c_1X+d_1)^{m_1}Z_1-(a_1X+b_1)^{m_1},\\
& H_2=(c_2Y+d_2)^{m_2}Z_2-(a_2Y+b_2)^{m_2},\\
&R_1(Z_1,Y)=\res_X(H_1,Q(X,Y)), R_2(Z_1,Z_2)=\res_Y(H_2,R_1(Z_1,Y)).
\end{align*} As $Q$ does not divide $H_1$, $R_1\neq 0$. Similarly, $R_2\neq 0$. Moreover, one can easily check that $R_2(\bar{\alpha},\bar{\beta})=0$. Hence $\bar{Q}$ divides $R_2$.
By Proposition~\ref{prop:heightproperty3} and Corollary~\ref{cor:resultant}, one has that
\begin{align*}
  \frakh(\bar{Q})\leq \frakh(R_2)&\leq \deg(R_1,Y)\frakh(H_2)+\deg(H_2,Y)\frakh(R_1)\\
  &\leq \deg(H_2,Y)(\deg(H_1,X)\frakh(Q)+\deg(Q,X)\frakh(H_1))\\
  &=\deg(H_2,Y)\deg(H_1,X)\frakh(Q)=m_1m_2\frakh(Q).
\end{align*}
Since $\divs(\bar{\alpha})^{-} \leq \divs(\bar{\beta})^{-}$, $\deg(\bar{Q},X)=\tdeg(\bar{Q})$ by the Proposition 2 in \cite{eremenko}.
If a place $\frakp$ is not a pole of $\bar{\beta}$ then it is not a pole of $\bar{\alpha}$ too. For such places, the lemma follows from Propositions~\ref{prop:specialcase},~\ref{prop:heightproperty1} and Corollary \ref{cor:heightproperty11}.
It remains to show that the case that $\frakp$ is a pole of $\bar{\beta}$. Suppose that $\frakp$ is a pole of $\bar{\beta}$. If $\frakp$ is also a pole of $\alpha$ then $\frakh(\pi_\frakp(\alpha))=0$ and there is nothing to prove. Assume that $\frakp$ is not a pole of $\alpha$.
If $\frakp$ is a pole of $\beta$ then $\pi_\frakp(\alpha)$ is a zero of $\tilde{Q}(X,0)$, where $\tilde{Q}=Y^rQ(X,1/Y)$ and $r$ is the smallest integer such that $\tilde{Q}\in K[X,Y]$, and thus $\frakh(\pi_\frakp(\alpha))\leq \frakh(\tilde{Q})=\frakh(Q)$ and we are done. Now suppose that $\frakp$ is not a pole of $\beta$. Since $\frakp$ is a pole of $\bar{\beta}$, $c_2\pi_\frakp(\beta)+d_2=0$, i.e. $\pi_\frakp(\beta)=-d_2/c_2$. Applying $\pi_\frakp$ to $Q(\alpha,\beta)=0$ yields that $\pi_\frakp(\alpha)$ is a solution of $Q(X,-d_2/c_2)=0$. Write $Q=\sum_{i=0}^\ell A_i(Y)X^i$ where $A_i(Y)=\sum_{j=0}^s a_{i,j}Y^j\in K[Y]$. Note that $A_i(-d_2/c_2)$ viewed as a polynomial in the $a_{i,j}$ is either 0 or homogeneous in the $a_{i,j}$ of degree 1 with coefficients in $\Q$. By Proposition~\ref{prop:heightproperty2},
\begin{align*}
   \frakh(Q(X,-d_2/c_2))&=\frakh((A_0(-d_2/c_2):\dots:A_\ell(-d_2/c_2)))\\
   &\leq \frakh((\dots:a_{i,j}:\cdots))=\frakh(Q).
\end{align*}
Hence $\frakh(\pi_\frakp(\alpha))\leq \frakh(Q(X,-d_2/c_2))\leq \frakh(Q)$ and the lemma holds.
\end{proof}

\begin{lemma}
\label{lm:distinctpoles}
Assume $S$ is a finite set of places of $L$ over $K$ and $\alpha\in L$.  Then there are $c_1,c_2\in \Q$ with $c_2\neq 0$ such that
$$
    \supp\left(\divs\left(\frac{\alpha}{c_1\alpha+c_2}\right)^{-}\right)\cap S =\emptyset.
$$
\end{lemma}
\begin{proof}
Set
$$
     M=\left\{ \pi_\frakp(\alpha) \mid  \mbox{$\forall\, \frakp\in S$ with $\ord_\frakp(\alpha)\geq 0$} \right\}. $$
Then $M$ is a finite subset of $K$. Let $c_1, c_2\in \Q$ satisfy that $c_2\neq 0$ and  $c_1a+c_2\neq 0$ for all $a\in M$.
For $\frakp\in S$ with $\ord_\frakp(\alpha)\geq 0$, one has that
$$
    \pi_\frakp(c_1\alpha+c_2)=c_1\pi_\frakp(\alpha)+c_2\neq 0,\, {\mathrm i.e.} \,\,\ord_\frakp(c_1\alpha+c_2)=0.
$$
This implies that $\ord_\frakp(\alpha/(c_1\alpha+c_2))=\ord_\frakp(\alpha)\geq 0$ for all $\frakp\in S$ with $\ord_\frakp(\alpha)\geq 0$. On the other hand, for $\frakp\in S$ with $\ord_\frakp(\alpha)<0$, one has that
$$
    \ord_\frakp(\alpha/(c_1\alpha+c_2))=\ord_\frakp(\alpha)-\ord_\frakp(c_1\alpha+c_2)=\ord_\frakp(\alpha)-\ord_\frakp(\alpha)=0.
$$
In either case, $\frakp$ is not a pole of $\alpha/(c_1\alpha+c_2)$. Thus $c_1,c_2$ have the desired property.
\end{proof}
\begin{lemma}
\label{lm:h(D)}Suppose that $P$ is a nonzero irreducible polynomial of total degree $\rho$ in $K[X,Y]$ satisfying $P(x,y)=0$ and $\bfd(X)$ is the discriminant of $P$ with respect to $Y$.  Let $\bar{x}=x/(c_1x+c_2)$ where $c_1,c_2\in \Q$ with $c_2\neq 0$. Then
$$\frakh(\pi_\frakp(x))\leq 2\rho\frakh(P)$$ for each
$\frakp\in \supp(\divs(x)^{-})\cup\supp(\divs(\bar{x})^{-})\cup \supp(\divs(y)^{-})\cup\supp(\divs(\bfd(x))).$
\end{lemma}
\begin{proof}
Suppose that $\pi_\frakp(x)=\infty$. Then $\frakh(\pi_\frakp(x))=0$ and the lemma is clear. In the following suppose that $\pi_\frakp(x)\neq \infty$.
Suppose that $\frakp\in \supp(\divs(\bar{x})^{-})$.
Then $\pi_\frakp(c_1x+c_2)=c_1\pi_\frakp(x)+c_2=0$ and so $\pi_\frakp(x)=-c_2/c_1$. Hence $\frakh(\pi_\frakp(x))=0$ and thus the lemma holds. Suppose that $\frakp\in \supp(\divs(y)^{-})$. Then $\pi_\frakp(y)=\infty$ and $(\pi_\frakp(x),0)$ is a zero of $\bar{P}=Y^rP(X,1/Y)$ where $r$ is the smallest integer such that $Y^rP(X,1/Y)\in K[X,Y]$. In other word, $\pi_\frakp(x)$ is a zero of $\bar{P}(X,0)$. Note that $\frakh(\bar{P})=\frakh(P)$. Hence $\frakh(\pi_\frakp(x))\leq \frakh(\bar{P}(X,0))\leq \frakh(P)\leq 2\rho\frakh(P)$.
Finally, suppose that $\supp(\divs(\bfd(x)))\neq \emptyset$, i.e. $\bfd(x)\notin K$ and $\frakp\in \supp(\divs(\bfd(x)))$. If $\frakp$ is a pole of $\bfd(x)$ then it is a pole of $x$ and we are already done. Suppose that $\frakp$ is a zero of $\bfd(x)$. Then $\pi_\frakp(\bfd(x))=0$ which implies that $\bfd(\pi_\frakp(x))=0$. Hence $\frakh(\pi_\frakp(x))\leq \frakh(\bfd(X))\leq 2\rho \frakh(P)$.
\end{proof}
Now we are ready to prove the main result of this paper.
\begin{theorem}
\label{thm:boundpoints}
Let $P$ be an irreducible polynomial in $K[X,Y]$ of degree $m$ with respect to $X$ and of degree $n$ with respect to $Y$.  Suppose that $\rho=\tdeg(P)$ and $0<\epsilon<1$.  Then for every $a,b\in K$ with $P(a,b)=0$, one has that
$$
  (1-\epsilon)n\frakh(b)-C\leq m\frakh(a)\leq (1+\epsilon)n\frakh(b)+C
$$
where
$$
C=75\cdot 2^{13}\cdot (1/\epsilon)^6(\rho+1)^{\frac{40(\rho+1)^9}{\epsilon^3}}\frakh(P).
$$
\end{theorem}
\begin{proof}
Let $L$ be the field of fractions of $K[X,Y]/(P)$. Then $L$ is an algebraic function field of one variable over $K$. Set $x=X+(P)$ and $y=Y+(P)$. Then $P(x,y)=0$. Choose $c_1,c_2\in \Q$ with $c_2\neq 0$ such that
$$\supp(\divs(x/(c_1x+c_2))^{-})\cap \supp(\divs(y)^{-})=\emptyset.$$
Such $c_1,c_2$ exist because of Lemma~\ref{lm:distinctpoles}. Let $\lambda_2$ be the smallest integer not less than $\frac{\rho}{2\epsilon}$ and let $\lambda_1$ be the largest integer not greater than $\lambda_2+\rho/2$. Then $\frac{\rho}{2\epsilon}\leq \lambda_2<\frac{\rho}{2\epsilon}+1$ and $\lambda_2+\rho/2-1<\lambda_1\leq \lambda_2+\rho/2$. These imply that
\begin{equation}
\label{eq:inequality3}\lambda_1-\lambda_2\geq \frac{\rho}{2}-1, \,\, \frac{\lambda_1}{\lambda_2}\leq 1+\epsilon,\,\,\lambda_1+\lambda_2\leq \frac{2(\rho+1)}{\epsilon}.
\end{equation}
Set $\bar{x}=x/(c_1x+c_2)$ and
$$
   D=\lambda_1 n\divs(y)^{-}-\lambda_2 m\divs(\bar{x})^{-}.
$$
Note that $\deg(\divs(y)^{-})=[L:K(y)]=m$ and $\deg(\divs(\bar{x})^{-})=[L:K(\bar{x})]=n$. One sees that
$$\deg(D)=\lambda_1 nm-\lambda_2 nm=(\lambda_1-\lambda_2) nm.$$
As $nm\geq n+m-1\geq \rho-1$ and $\lambda_1-\lambda_2\geq \rho/2-1$, $\deg(D)\geq (\rho-1)(\rho-2)/2$ is not less than the genus of $P(X,Y)=0$. Consequently, $\calL_K(D)\neq \{0\}$. Let $\delta_D, h(D)$ be as in Notation~\ref{notation:divisor}. Then $h(D)\leq 2\rho\frakh(P)$ by Lemma~\ref{lm:h(D)} and
by (\ref{eq:inequality3}), $\delta_D=(\lambda_1+\lambda_2)nm< 2(1+\rho)\rho^2/\epsilon$.

Due to Proposition~\ref{prop:heightsofRiemann-Roch}, $\calL_K(D)$ contains a nonzero element $z=g(\bfa)/q(x)$, i.e. $\divs(z)+D\geq 0$, where $g(\bfa)$ is of the form (\ref{eq:riemann-rochelement}).
As $\supp(\divs(\bar{x})^{-})\cap \supp(\divs(y)^{-})=\emptyset$, one sees that
$$\divs(z)^{-}\leq \lambda_1 n\divs(y)^{-}, \,\, \lambda_2 m\divs(\bar{x})^{-}\leq \divs(z)^{+}=\divs(1/z)^{-}.$$ Suppose that $Q_1\in K[X,Z],Q_2\in K[Y,Z]$ are nonzero irreducible polynomials such that $Q_1(x,z)=0$ and $Q_2(y,z)=0$. By Proposition~\ref{prop:equality}, one has that
$$\frakh(Q_1), \frakh(Q_2)\leq 1600(\rho+\delta_D)^6(\rho+1)^{5(\rho+\delta_D)^3-9}h(D)\triangleq T.$$
Let $\frakp$ be a place of $L$ over $K$ such that $\pi_\frakp(x)=a$ and $\pi_\frakp(y)=b$. By Lemma~\ref{lm:pointinequality},
\begin{align*}
  \frakh(\pi_\frakp(z))&\leq \lambda_1 n\frakh(b)+\lambda_1 n\frakh(Q_2),\\
  \lambda_2 m\frakh(a)&\leq \frakh(\pi_\frakp(z))+\lambda_2 m \frakh(Q_1).
\end{align*}
The above two inequalities imply that
$$
    \lambda_2 m\frakh(a)\leq \lambda_1 n\frakh(b)+\lambda_1 n\frakh(Q_2)+\lambda_2 m \frakh(Q_1).
$$
In other words, $m\frakh(a)\leq (\lambda_1/\lambda_2)n\frakh(b)+(m+n\lambda_1/\lambda_2)T.$ Note that by (\ref{eq:inequality3}) $\lambda_1/\lambda_2\leq 1+\epsilon\leq 2$. One has that
\begin{equation}
\label{eq:righthandside}
   m\frakh(a)\leq (1+\epsilon)n\frakh(b)+3\rho T.
\end{equation}
Note that $\rho+\delta_D\leq \rho+2(\rho+1)\rho^2/\epsilon<2(\rho+1)^3/\epsilon$. One sees that
\begin{align*}
 3\rho T&= 3\rho\times 1600(\rho+\delta_D)^6(\rho+1)^{5(\rho+\delta_D)^3-9}h(D)\\
 &\leq  9600\rho^2(2(\rho+1)^3/\epsilon)^6(\rho+1)^{5(\rho+2(\rho+1)\rho^2/\epsilon)^3-9}\frakh(P)\\
 &\leq 9600\cdot 2^6\cdot (1/\epsilon)^6(\rho+1)^{5(\rho+2(\rho+1)\rho^2/\epsilon)^3+11}\frakh(P)\\
 &\leq 75\cdot 2^{13}\cdot (1/\epsilon)^6(\rho+1)^{\frac{40(\rho+1)^9}{\epsilon^3}}\frakh(P) \triangleq C.
\end{align*}
The last inequality holds because
$$5(\rho+2(\rho+1)\rho^2/\epsilon)^3+11<5(\rho+2(\rho+1)\rho^2/\epsilon+2)^3<40(\rho+1)^9(1/\epsilon)^3.$$
Set $\tilde{D}=\lambda_2 m \divs(\bar{x})^{-}-(2\lambda_2-\lambda_1)n\divs(y)^{-}$. Then
\begin{align*}
   \lambda_2-(2\lambda_2-\lambda_1)&=\lambda_1-\lambda_2\geq \frac{\rho}{2}-1, \\
   \frac{2\lambda_2-\lambda_1}{\lambda_2}&=2-\frac{\lambda_1}{\lambda_2}\geq 1-\epsilon,\\
   \lambda_2+2\lambda_2-\lambda_1&=3\lambda_2-\lambda_1\leq 2\lambda_2-\frac{\rho}{2}+1<\frac{2(\rho+1)}{\epsilon}.
\end{align*}
Using a similar argument, one has that
\begin{equation}
\label{eq:lefthandside}
   (1-\epsilon)n\frakh(b)\leq m\frakh(a)+C.
\end{equation}
Combining (\ref{eq:lefthandside}) with (\ref{eq:righthandside}) yields the conclusion.
\end{proof}

\bibliographystyle{plain}
\bibliography{symsolutions}
\end{document}